\documentclass[letterpaper, 10 pt, onecolumn]{ieeeconf}  % Comment this line out if you need a4paper
\usepackage{moreverb}
\usepackage{xcolor}
\usepackage{bm}
\usepackage{subfig}
\usepackage{amsmath}
\usepackage{mathptmx} % assumes new font selection scheme installed
\usepackage{amssymb}  % assumes amsmath package installed
\usepackage{epstopdf}
\usepackage{booktabs}
\usepackage[]{graphicx}

\newtheorem{remark}{Remark}
\newtheorem{assumption}{Assumption}
\newtheorem{definition}{Definition}
\newtheorem{proposition}{Proposition}
\newtheorem{property}{Property}
\newtheorem{theorem}{\textbf{Theorem}}
\newtheorem{corollary}{Corollary}
\newtheorem{lemma}{Lemma}

\IEEEoverridecommandlockouts                              % This command is only needed if
\newcommand{\sss}[1]{_{\scriptscriptstyle #1}}

\newcommand\BibTeX{{\rmfamily B\kern-.05em \textsc{i\kern-.025em b}\kern-.08em
		T\kern-.1667em\lower.7ex\hbox{E}\kern-.125emX}}

\newcommand{\rev}[1]{#1}

\title{~Learning Model Predictive Control with Long Short-Term Memory Networks}

\author{Enrico Terzi$^{*}$, Fabio Bonassi, Marcello Farina, and Riccardo Scattolini
	\thanks{$^{*}$ corresponding author. The authors are with the Dipartimento di Elettronica, Informazione e Bioingegneria, Politecnico di Milano, Via Ponzio 34/5, 20133, Milano, Italy. E-mail: {\tt\small name.surname@polimi.it}}}

\begin{document}
	
	\maketitle
	\thispagestyle{empty}
	\pagestyle{empty}

\begin{abstract}
	This paper analyzes the stability-related properties of Long Short-Term Memory (LSTM) networks and investigates their use as the model of the plant in the design of Model Predictive Controllers (MPC).
	First, sufficient conditions guaranteeing the Input-to-State stability (ISS) and Incremental Input-to-State stability ($\delta$ISS) of LSTM are derived.
	These properties are then exploited to design an observer with guaranteed convergence of the state estimate to the true one. 
	Such observer is then embedded in a MPC scheme solving the tracking problem. The resulting closed-loop scheme is proved to be asymptotically stable.
	The training algorithm and control scheme are tested numerically on the simulator of a \emph{pH} reactor, and the reported results confirm the effectiveness of the proposed approach.	
	
	\end{abstract}

	\maketitle

\section{Introduction}

The availability of large and informative datasets, collected on plants during long periods of time and spanning many different working conditions, is nowadays a typical starting point in control-related projects \cite{hou2013model, hand2006data}. Also thanks to the recent introduction and popularity of novel tools and algorithms for extracting information from data \cite{wu2008top}, engineers and scientists are increasingly focusing on data-based identification and control techniques \cite{bristow2006survey, aswani2013provably}.
Several approaches are aimed at the \textit{direct} learning of the controller from data \cite{tanaskovic2017data}, these algorithms can be either based on a - possibly reference - model, like in the Virtual Reference Feedback Tuning approach \cite{CAMPI20021337} and in Iterative Learning \cite{bristow2006survey}, or exploit model-free techniques, like Reinforcement Learning \cite{lenz2015deepmpc}.
On the other hand, \textit{indirect} approaches are aimed at first finding a model of the plant, based on which the controller is designed.
In the latter category, a quite recent model class that has gained extraordinary attention and popularity is the class of \textit{Neural Networks} (NN) \cite{haykin1994neural}, which have proven to be effective in a large variety of contexts and tasks, like image \cite{krizhevsky2012imagenet}, speech \cite{graves2013speech}, and handwriting recognition \cite{graves2009offline}, prediction \cite{zhang2003time}, and forecasting \cite{zhang1998forecasting,hippert2001neural}.

In the control context, and in order to account for the dynamic nature of the systems to be controlled, recurrent Neural Networks (RNN) have already been studied \cite{miller1995neural,delgado1995dynamic} and used in a number of applications \cite{wong2018recurrent,li2016distributed,jin2018robot,lanzetti2019recurrent}. In RNN, the output of the network is fed-back as input, so constituting a loop which allows to properly describe the dynamics of the system. However, the tuning of RNN calls for a complex training algorithm, that is affected by the so-called ``vanishing (or exploding) gradient" problem \cite{hochreiter1998vanishing}. Essentially, this prevents a proper training given the recursive equations featuring the network, that cause a vanish (or explosion) of information and gradient over the iterations. Up to date, only a couple of architectures proved to be able to practically overcome this issue, namely Echo State Networks (ESN) \cite{jaeger2002tutorial} and Long Short Term Memory (LSTM) networks \cite{greff2016lstm}.

Despite their potential impact in the control field, the literature regarding the theoretical properties of RNNs is little, since they are often tested empirically and with no theoretical guarantees, in favour of experimental evidence. 
This represents a strong limitation on the use of RNN in the design of control systems, and motivates the analysis of their properties from a control-theoretical perspective. For all these reasons, the properties of ESN, in terms of stability, and their use as suitable models of the plant in the design of Model Predictive Control (MPC) regulators with stability guarantees, have been recently analyzed \cite{8728229}.
Although ESsN have proven to be effective and characterized by a simple training procedure, LSTM \cite{gers2000learning} are gaining a wider popularity. First introduced in 1997 \cite{hochreiter1997long}, LSTM are nowadays widely used for several tasks \cite{xingjian2015convolutional, sundermeyer2012lstm} and in everyday's devices, such as mobile phones and GPS navigators for speech recognition. This diffusion is due to their flexibility and ability to recover long-term dependencies across the data thanks to their internal states. In the context of dynamical systems and control, some very recent stability results about their autonomous, i.e. non forced, version have been described by Deka et al. \cite{Deka2019, 8702629}, and  Amrouche et al. \cite{amrouche2018long}, where also an analysis of their equilibria has been reported.

In this paper, we investigate the use of LSTM in the context of Model Predictive Control, extending some preliminary investigations \cite{bonassi2019lstm}.
First, conditions on LSTM's parameters (internal weights) guaranteeing the Input-to-State Stability\cite{ljung2001system} (ISS) and the Incremental Input-to-State Stability\cite{Discrete-timedeltaISS} ($\delta$ISS) properties are derived; 
notably, these conditions explicitly depend on the model parameters and can be forced in the training phase of the network.
 Then, assuming that the trained net exactly represents the model of the system, and relying on $\delta$ISS, we design an observer guaranteeing that the estimated state asymptotically converges to the true value. 
 Based on the LSTM model and on the state observer, we then design an MPC control algorithm solving the tracking problem for constant state and input reference values and in presence of input constraints. The stability of the equilibrium point is obtained with a suitable tuning of the MPC design parameters, i.e. the state and control weighting functions and the terminal cost. Notably, no terminal constraints are required, which makes the tuning procedure easier.
\rev{A probabilistic method, based on the Scenario Approach \cite{bonassi2019lstm, campi2009scenario}, is finally proposed to estimate an upper bound of the model-plant mismatch, which is typically required to design robust MPC algorithms coping with the possible presence of model uncertainties.}

The performances of the overall control system  are tested numerically on the simulator of a \emph{pH} neutralization process \cite{hall1989modelling}, that represents a well-known benchmark for nonlinear SISO systems. The modeling capability of the trained LSTM is first quantitatively evaluated on a validation dataset, then closed-loop experiments are reported, witnessing the potentialities of the proposed approach.

The paper is organized as follows: in Section \ref{sec:model} the dynamic model of LSTM is analyzed, and the conditions guaranteeing the ISS and the $\delta$ISS properties are established. In Section \ref{sec:control} the design of the stabilizing observer and of the MPC algorithm is discussed, while in Section \ref{sec:example} the numerical example is described. Finally, conclusions and hints for future work are included in Section \ref{sec:conclusions}. An Appendix reports the proofs of the theoretical results.\smallskip\\
\textbf{Notation and basic definitions}.\\
We denote $v_{(j)}$ the $j$-th entry of vector $v$. $0_{a,b}$ is the null matrix of dimension $a\times b$, $I_n$ is the identity matrix of order $n$. Moreover, given a vector $v$, we denote $\|v\|$ as the 2-norm of $v$, $\|v\|^2_A=v^TAv$ its squared norm weighted by matrix $A$, and with $\|v\|_{\infty}$ its infinity norm, i.e., $\|v\|_{\infty}=\max_{j=1,\dots,n}|v_{(j)}|$, being $n$ the number of entries of $v$. $v^T$ denotes vector transpose and diag$(v)$ the diagonal matrix with $v$ on the diagonal. We denote with $\|A\|$ and with $\|A\|_{\infty}$ the induced 2-norm and $\infty$-norm of $A$, respectively, while $\rho(M)$ is the spectral radius of the square matrix $M$ (i.e. maximum absolute value of its eigenvalues). Given an interval $[a,b] \subset \mathbb{R}$ and a positive integer $n$ we denote $[a,b]^n=\{x \in \mathbb{R}^n:x_{(j)} \in [a,b], \forall j=1,\dots,n\}$. The same notation is applied for open intervals.
With reference to the discrete-time system
\begin{equation}
	\label{eq:intro:discretetime}
	\chi^+=\varphi(\chi,u),
\end{equation}
where $\chi$ is the state vector, $u$ is the input vector, and $\varphi(\cdot)$ is a nonlinear function of the input and the state, $\chi^+$ indicates the value of $\chi$ at the next time step. We indicate with $\chi_{i}(k)$ the solution to system \eqref{eq:intro:discretetime} at time step $k$ starting from the initial state $\chi_{0i}$ with input sequence $u_i(0)$, $\dots$, $u_i(k-1)$. For the sake of readability, time index $k$ will be omitted where possible and clear from the context.
Let us now recall some definition, see \cite{Discrete-timedeltaISS}, useful for the following developments.
\begin{definition} [$\mathcal{K}$-Function]
	A continuous function $\alpha: \mathbb{R}_{\geq0}\to\mathbb{R}_{\geq0}$ is a class $\mathcal{K}$ function if $\alpha(s)>0$ for all $s>0$, it is strictly increasing, and $\alpha(0)=0$.
\end{definition}
\begin{definition} [$\mathcal{K}_{\infty}$-Function]
	A continuous function $\alpha: \mathbb{R}_{\geq0}\to\mathbb{R}_{\geq0}$ is a class $\mathcal{K}_{\infty}$ function if it is a class $\mathcal{K}$ function and $\alpha(s)\to\infty$ for $s\to\infty$.
\end{definition}
\begin{definition} [$\mathcal{KL}$-Function]\label{def:KL}
	A continuous function $\beta: \mathbb{R}_{\geq0}\times\mathbb{Z}_{\geq0}\to\mathbb{R}_{\geq0}$ is a class $\mathcal{KL}$ function if $\beta(s,k)$ is a class $\mathcal{K}$ function with respect to $s$ for all $k$, it is strictly decreasing in $k$ for all $s>0$, and $\beta(s,k)\to0$ as $k\to\infty$ for all $s>0$.
\end{definition}
\rev{
\begin{definition}[ISS\cite{bonassi2019lstm, jiang2001input}]\linespread{1.2}
	\label{def:ISS}
	System \eqref{eq:intro:discretetime} is called \textit{input-to-state stable} in $\mathcal{X}$ with respect to $\mathcal{U}$, if there exist functions $\beta\in\mathcal{KL}$ and $\gamma_v, \gamma_b\in \mathcal{K}_{\infty}$ such that, for any $k\in\mathbb{Z}_{\geq0}$, any initial state $\chi_{0} \in \mathcal{X}$, any input sequence $\{v(0), v(1), ...: v(\tau) \in \mathcal{U} \}$, and any bias $b_c$, it holds that:
	\begin{equation}\vspace{-0.2cm}	\| \chi(k) \|\leq\beta(\|\chi_{01} \|,k)+\gamma_v(\max_{h\geq 0}\|v(h)\|) + \gamma_b(\| b_c \|)\\[1.5ex]
	\label{eq:ISS}
	\end{equation}\normalsize
\end{definition}
}
\begin{definition}[$\delta$ISS\cite{Discrete-timedeltaISS}]\linespread{1.2}
	\label{def:deltaISS}
	System \eqref{eq:intro:discretetime} is called \textit{incrementally input-to-state stable} in $\mathcal{X}$ with respect to $\mathcal{U}$, if there exist functions $\beta_\delta\in\mathcal{KL}$ and $\gamma_\delta\in \mathcal{K}_{\infty}$ such that, for any $k\in\mathbb{Z}_{\geq0}$, any initial states $\chi_{01},\chi_{02}\in \mathcal{X}$, and any pair of input sequences $\{v_1(0), v_1(1), ...: v_1(\tau) \in \mathcal{U} \}$ and $\{v_2(0), v_2(1), ...: v_2(\tau) \in \mathcal{U} \}$, it holds that:
	\begin{equation}\vspace{-0.2cm}	\|\chi_1(k)-\chi_2(k)\|\leq\beta_\delta(\|\chi_{01}-\chi_{02}\|,k)+\gamma_\delta(\max_{h\geq 0}\|v_{1}(h)-v_{2}(h)\|)\\[1.5ex]
	\label{eq:deltaISS}
	\end{equation}\normalsize
\end{definition}\linespread{1}

\rev{The ISS property guarantees the vanishing contribution of initialization and the boundedness of the state trajectories, and allows in this work to perform the safety verification of the network \cite{bonassi2019lstm}. 
On the other hand, the $\delta$ISS property is commonly required for many purposes, e.g. to design MPC regulators \cite{Discrete-timedeltaISS, kohler2019simple} and Moving Horizon estimators \cite{rao2003constrained, alessandri2008moving}, as it guarantees that the effects of different initializations vanish, and that feeding the network with two different input sequences leads to state trajectories with bounded distance. In the following we devise conditions under which the LSTM \eqref{eq:model:lstm} are guaranteed to feature these properties.}

%
%\begin{definition}[$\delta$ISS Lyapunov function\cite{Discrete-timedeltaISS}]
%	A function $V:\mathbb{R}^{n}\times\mathbb{R}^{n}\to\mathbb{R}_{+}$ is called a $\delta$ISS-Lyapunov function for \eqref{eq:intro:discretetime}, if there exist functions $\xi_{1},\xi_{2},\xi\in\mathcal{K}_{\infty}$ and $\sigma\in\mathcal{K}$ so that, for all $x_{1},x_{2}\in\mathbb{R}^{n}$ and $\omega_{1},\omega_{2}\in\mathbb{R}^{m}$, it holds that:
%	\begin{equation}
%	\label{firstCond1}
%	\xi_{1}(\|x_{1}-x_{2}\|)\leq V(x_1,x_2)\leq\xi_{2}(\|x_{1}-x_{2}\|)
%	\end{equation}
%	\begin{equation}
%	\label{secondCond1}
%	\begin{split}
%	&V(\varphi(x_{1},\omega_{1}),\varphi(x_{2},\omega_{2}))-V(x_{1},x_{2})\\[1ex] &\leq-\xi(\|x_{1}-x_{2}\|)+\sigma(\|\omega_{1}-\omega_{2}\|)
%	\end{split}
%	\end{equation}
%\end{definition}
%\begin{property}[\cite{demmel1997applied}]
%	\label{property1}
%	Given a matrix $A$, it holds that $\|A\| \geq \rho(A)$. In particular, if $A$ is symmetric, then $\|A\| = \rho(A)$.
%\end{property}
%We recall the following results.
%\begin{theorem}[\cite{Discrete-timedeltaISS}]
%	\label{deltaISS NC}
%	If system \eqref{eq:intro:discretetime} admits a time invariant $\delta$ISS Lyapunov function, then it is $\delta$ISS.
%\end{theorem}
%
%\begin{theorem}[\cite{Discrete-timedeltaISS}]
%	\label{deltaISS-ISS}
%	The class of $\delta$ISS systems is a strict subset of $ISS$ systems.
%\end{theorem}
%
%
\section{LSTM networks}\label{sec:model}
\subsection{State space form}
The LSTM network, with input $u \in \mathbb{R}^{n_u}$ and output $y \in \mathbb{R}^{n_y}$, is described by the following system of equations \cite{gers2001lstm,gers2002learning}.
\begin{subequations}\label{eq:model:lstm}
	\begin{align}
	x^+	 =& \, \sigma_g(W_f u + U_f \xi + b_f) \circ x +\sigma_g(W_i u + U_i \xi + b_i) \circ \sigma_c(W_c u + U_c \xi + b_c) \label{eq:model:lstm_x} \\
	\xi^+ =& \, \sigma_g(W_o u + U_o \xi + b_o) \circ \sigma_c(x^+) 	\label{eq:model:lstm_xi} \\
	y =& \, C\xi + b_y \label{eq:model:lstm_output}
	\end{align}
\end{subequations}
The vector $\chi=\begin{bmatrix}x^T & \xi^T\end{bmatrix}^T$ is the state of the network, so that \eqref{eq:model:lstm} can be rewritten in the general form \eqref{eq:intro:discretetime}. In the related terminology, $x \in \mathbb{R}^{n_x}$ is named \emph{hidden state}, while $\xi \in \mathbb{R}^{n_x}$ is named \emph{output state} (or cell).\\
In system \eqref{eq:model:lstm}, $\sigma_g(x)=\frac{1}{1+e^{-x}}$ and $\sigma_c(x)=\text{tanh}(x)$; when applied to a vector, we assume to apply them entry-wise. Also, $\circ$ is the element-wise (Hadamard) product. The terms $W_f,W_i,W_o,W_c \in \mathbb{R}^{n_x \times n_u}$, $U_f,U_i,U_o,U_c \in \mathbb{R}^{n_x \times n_x}$,$C \in \mathbb{R}^{n_y\times n_x}$ are weighting matrices and $b_f,b_i,b_o,b_c \in \mathbb{R}^{n_x}, b_y \in \mathbb{R}^{n_y}$ are biasing vectors.

\begin{assumption}[Boundedness of $u$] \label{ass:u_bounded}
	The input is bounded, i.e.
	\begin{equation}\label{eq:model:bound_u}
		u \in \mathcal{U}=[-u_{\scriptscriptstyle \text{max}}, u_{\scriptscriptstyle \text{max}}]^{n_u}.
	\end{equation}
\end{assumption}
Note that Assumption \ref{ass:u_bounded} is quite general. It could be associated to physical saturations of the input variable or can be achieved by means of a proper normalization of the dataset employed for training \cite{goodfellow2016deep}.
\rev{\begin{remark}
	In the LSTM model \eqref{eq:model:lstm}, the logistic and tanh activation functions have been considered. However, the proposed theory can be readily extended to generic monotonically increasing upper- and lower-bounded  functions, provided that $\sigma_{c}(0) = 0$.
\end{remark}}

\subsection{Properties of the system functions and bounds on the variables}
First of all, note that, in view of their definitions,
\begin{subequations} \label{eq:model:bounds_sigmas}
\begin{align}
&\sigma_g(t) \in (0,1), \quad \forall t \in \mathbb{R} \label{eq:bound_sigmag}\\
&\sigma_c(t) \in (-1,1), \quad \forall t \in \mathbb{R}\label{eq:bound_sigmac},
\end{align}
\end{subequations}
Also, $\sigma_g(t)$ and $\sigma_c(t)$ are Lipschitz continuous functions \cite{sohrab2003basic} with Lipschitz constants $L_g=0.25$ and $L_c=1$, respectively, and they are both strictly monotonic.
In view of \eqref{eq:model:bounds_sigmas}, see \eqref{eq:model:lstm},
\begin{align} \label{eq:model:bound_xi}
\xi\in (-1,1)^{n_x}, \,\, \text{ i.e. } \, \xi_{(j)} \in (-1,1), \forall j=1,\dots,n_x.
\end{align}
Rewriting equation \eqref{eq:model:lstm} for each entry of the state vectors we obtain:
\begin{subequations}
\begin{align}
x_{(j)}^+	 = & \,	\sigma_{g}(W_f u + U_f \xi + b_f)_{(j)} \circ x_{(j)}+ \sigma_{g}(W_iu + U_i \xi + b_i)_{(j)} \circ \sigma_{c}(W_cu + U_c \xi + b_c)_{(j)} \label{eq:model:lstm_i}\\
\xi_{(j)}^+	 =& \, \sigma_{g}(W_o u + U_o \xi + b_o)_{(j)} \circ \sigma_{c}(x^+)_{(j)} 		\label{eq:model:output_i}
\end{align}
\end{subequations}
Note that, in \eqref{eq:model:lstm_i}, for each $j\in 1,\dots,n_x$,
\begin{equation}\label{eq:model:bound:sigma_g_f}
\begin{aligned}
\left\lvert \sigma_{g}(W_f u + U_f \xi + b_f)_{(j)}\right\lvert & \leq \left\|\sigma_{g}(W_f u + U_f \xi + b_f)\right\|_{\infty} \leq \max_{u \in \mathcal{U},\xi\in (-1,1)^{n_x}} \left\| \sigma_{g}(W_f u + U_f \xi + b_f)\right\|_{\infty}  \\
& \leq \left\| \max_{u \in \mathcal{U},\xi\in (-1,1)^{n_x}} \sigma_{g}(W_f u + U_f \xi + b_f) \right\|_{\infty}  \leq \left\| \sigma_{g} \left(\max_{u \in \mathcal{U},\xi\in (-1,1)^{n_x}} [W_f u + U_f \xi + b_f] \right) \right\|_{\infty} \\
& \leq  \sigma_{g} \left(\max_{u \in \mathcal{U},\xi\in (-1,1)^{n_x}} \|W_f u + U_f \xi + b_f\|_{\infty} \right) \leq  \sigma_{g} \left(\|\begin{bmatrix}W_f u_{\scriptscriptstyle \text{max}} & U_f & b_f \end{bmatrix}\|_{\infty} \right) = \bar{\sigma}_g^f 
\end{aligned}
\end{equation}
where we relied on \eqref{eq:model:bound_u} and \eqref{eq:model:bound_xi}.
With similar arguments we derive:
\begin{align}
\left\lvert  \sigma_{g}(W_iu + U_i \xi + b_i)_{(j)} \right\lvert &\leq\bar{\sigma}_g^i= \sigma_{g}(\|\begin{bmatrix}W_i u_{\scriptscriptstyle \text{max}} & U_i & b_i \end{bmatrix}\|_{\infty})\\
\left\lvert \sigma_{g}(W_o u + U_o \xi + b_o)_{(j)} \right\lvert &\leq\bar{\sigma}_g^o= \sigma_{g}(\|\begin{bmatrix}W_o u_{\scriptscriptstyle \text{max}} & U_o & b_o \end{bmatrix}\|_{\infty}) \label{eq:model:bound:sigma_g_o}\\
\left\lvert  \sigma_{c}(W_cu + U_c \xi + b_c)_{(j)} \right\lvert &\leq \bar{\sigma}_c^c= \sigma_{c}(\|\begin{bmatrix}W_c u_{\scriptscriptstyle \text{max}} & U_c & b_c \end{bmatrix}\|_{\infty}) \label{eq:model:bound:sigma_c_c}
\end{align}

%\begin{figure}
%	\centering
%	\includegraphics[scale=0.35]{LSTMcell.png}
%	\caption{Schematic representation of a LSTM cell}
%	\label{fig_LSTMcell}
%\end{figure}
%
%\begin{figure}
%	\centering
%	\includegraphics[scale=0.5]{LSTM_net.jpg}
%	\caption{Representation of a LSTM network with a LSTM % layer and a linear output layer, the latter with weights (entries of matrix $C$) represented with white squares}
%	\label{fig_LSTMnet}
%\end{figure}
%
Also, by analyzing equation \eqref{eq:model:lstm_i}, and recalling \eqref{eq:model:bound:sigma_g_f}-\eqref{eq:model:bound:sigma_c_c}, we define an invariant set $\mathcal{X}=\bigg\{x \in \mathbb{R}:|x|\leq \frac{\bar{\sigma}_g^i\bar{\sigma}_c^c}{1-\bar{\sigma}_g^f}\bigg\}$ for $x_{(j)}$, i.e. such that
\begin{equation} \label{eq:model:bound_x}
|x_{(j)}(0)| \in \mathcal{X} \implies |x_{(j)}(t)| \in \mathcal{X}, \quad \forall t \geq 0.
\end{equation}
Thanks to this definition, we can bound $\sigma_c(x^+)_{(j)}$ in \eqref{eq:model:output_i}, namely
\begin{equation}\label{eq:model:bound:sigma_c_x}
|\sigma_c(x^+)_{(j)}| \leq \bigg|\sigma_c\left(\frac{\bar{\sigma}_g^i\bar{\sigma}_c^c}{1-\bar{\sigma}_g^f}\right)\bigg|\leq \sigma_c \left(\frac{\bar{\sigma}_g^i\bar{\sigma}_c^c}{1-\bar{\sigma}_g^f}\right)= \bar{\sigma}_c^x
\end{equation}

\subsection{Stability properties of the LSTM networks}

\rev{
In the following, sufficient conditions guaranteeing the stability properties are presented. For compactness, all the proofs are reported in the Appendix.
\begin{theorem} \label{thm:iss}
	The LSTM network \eqref{eq:model:lstm} is ISS with respect to the input $u$ and bias $b_c$ if $\rho(A) < 1$, where 
	\begin{equation} \label{eq:iss:A}
		A = \begin{bmatrix}
			\bar{\sigma}_g^f & \bar{\sigma}_g^i \| U_c \| \\
			\bar{\sigma}_g^o \bar{\sigma}_g^f  & \bar{\sigma}_g^o \bar{\sigma}_g^i \| U_c \|
		\end{bmatrix}.
	\end{equation}
\end{theorem}
\begin{proposition}\label{prop:iss:schur}
	The Schur stability of the matrix $A$ defined in \eqref{eq:iss:A} is ensured if the following inequality holds:
	\begin{equation}\label{eq:iss:condition}
	\bar{\sigma}_g^f + \bar{\sigma}_g^o \bar{\sigma}_g^i \| U_c \| < 1.
	\end{equation} \smallskip
\end{proposition}
}

\begin{theorem}\label{thm:deltaISS}
	Denoting $$ \alpha= \frac{1}{4}\|U_f\|   \frac{ \bar{\sigma}_g^i\bar{\sigma}_c^c}{1-\bar{\sigma}_g^f}+   \bar{\sigma}_g^i   \|U_c\|   + \frac{1}{4} \|U_i\|    \bar{\sigma}_c^c    ,$$
	the LSTM network \eqref{eq:model:lstm} is $\delta$ISS with respect to the inputs $u_1$ and $u_2$ if $\rho(A_\delta)<1$, where
	\begin{equation}\label{eq:deltaiss:A}
	A_\delta=\begin{bmatrix}
	\bar{\sigma}_g^f & \alpha\\
	\bar{\sigma}_g^o\bar{\sigma}_g^f & \,\,\,  \alpha\bar{\sigma}_g^o + \frac{1}{4}\bar{\sigma}_c^x\|U_o\|\end{bmatrix}.
	\end{equation}
\end{theorem}

\begin{proposition}\label{prop:deltaiss:schur}
	The Schur stability of the matrix $A_\delta$ defined in \eqref{eq:deltaiss:A} is ensured if the following inequalities hold:
	\begin{equation} \label{eq:deltaiss:condition}
	-1+ \bar{\sigma}_g^f +\alpha\bar{\sigma}_g^o + \frac{1}{4}\bar{\sigma}_c^x\|U_o\|  <\frac{1}{4}\bar{\sigma}_g^f\bar{\sigma}_c^x\|U_o\| < 1.
	\end{equation} 
\end{proposition}

\rev{It is worth noting that the $\delta$ISS property \cite{Discrete-timedeltaISS} implies the ISS property \cite{jiang2001input}. 
Indeed, this relationship also holds for the sufficient criteria of Proposition \ref{prop:iss:schur} and Proposition \ref{prop:deltaiss:schur}, as stated in Corollary \ref{prop:deltaiss_implies_iss}.
\begin{corollary} \label{prop:deltaiss_implies_iss}
	Given the LSTM network \eqref{eq:model:lstm}, the satisfaction of the $\delta$ISS sufficient condition \eqref{eq:deltaiss:condition} in  Proposition \ref{prop:deltaiss:schur} implies the satisfaction of the ISS sufficient condition \eqref{eq:iss:condition} in Proposition \ref{prop:iss:schur}.
\end{corollary}}

\begin{remark} \label{remark:train}
The conditions of Proposition \ref{prop:iss:schur} and Proposition \ref{prop:deltaiss:schur} are explicit functions of the LSTM parameters. 
\rev{They can be either checked to a-posteriori verify the ISS and $\delta$ISS properties of the trained network, or they can be enforced during the training procedure. In the latter case, depending on the training algorithm, they can be stated as hard nonlinear constraints, or they can be relaxed by moving the constraint residual to the loss function of the training algorithm \cite{goodfellow2016deep}. Additional details on this procedure are reported in Section \ref{sec:example}.}
\end{remark}

\section{Control design} \label{sec:control}

\subsection{Observer design} \label{sec:obsv}
\begin{figure}
	\centering
	\includegraphics[scale=0.5]{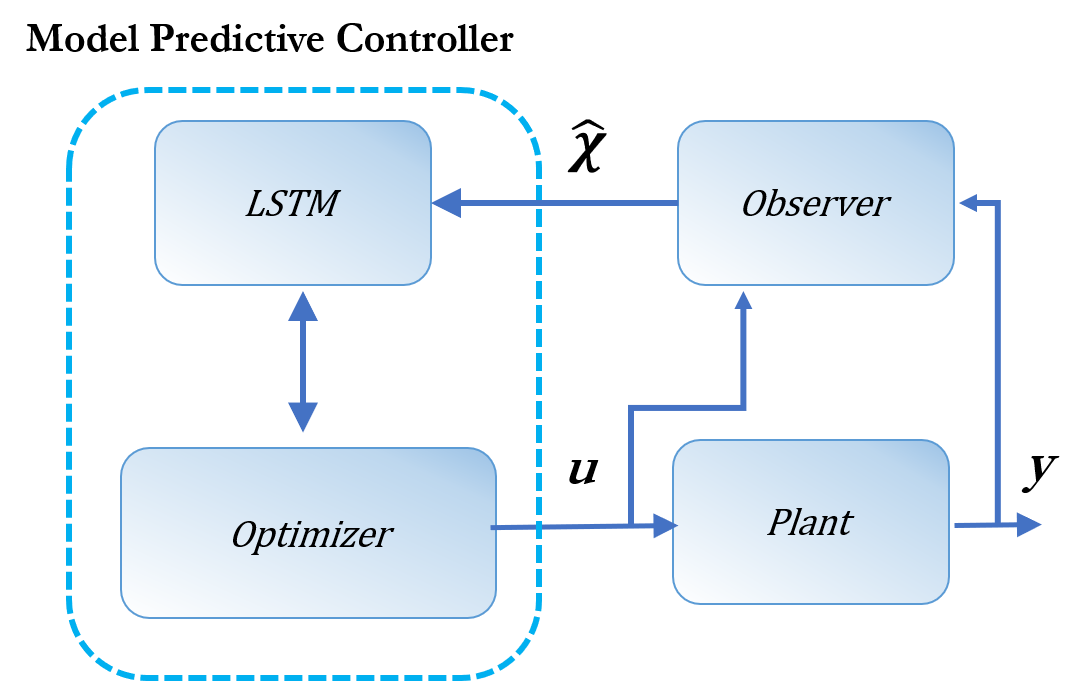}
	\caption{Control architecture}
	\label{fig:control_scheme}
\end{figure}
The use of the LSTM network for model predictive control purposes calls for the availability of a state estimate of the plant, as represented in Figure \ref{fig:control_scheme}. We propose the use of a tailored observer, guaranteeing a fast convergence of such estimate to the real state value.
The observer is a dynamical system with state $ \hat{\chi}= \begin{bmatrix}\hat{x}^T & \hat{\xi}^T \end{bmatrix}^T $ and output estimate $\hat{y}$, taking the following form:
\begin{equation} \label{eq:obsv:model}
\begin{aligned}
	\hat{x}^+	 =& 	\sigma_g[W_f u + U_f \hat{\xi} + b_f + L_f(y-\hat{y}) ] \circ \hat{x} + \sigma_g[W_i u + U_i \hat{\xi} + b_i + L_i(y-\hat{y}) ] \circ \sigma_c(W_c u + U_c \hat{\xi} + b_c) \\
	\hat{\xi}^+	 =&   \sigma_g[W_o u + U_o \hat{\xi} + b_o + L_o(y-\hat{y}) ] \circ \sigma_c(\hat{x}^+) \\
	\hat{y}      =&   C\hat{\xi} + b_y 
\end{aligned}
\end{equation}
where $L_f,L_i$ and $L_o \in \mathbb{R}^{n_x \times n_y}$ are suitable observer gains to be properly selected. 
In the following, theoretical results concerning the design of the state observer are reported, while the corresponding proofs are detailed in the Appendix.
\begin{theorem}\label{thm:obsv_convergence}
If the plant behaves according to \eqref{eq:model:lstm} and $\rho(A_\delta) < 1$, the observer \eqref{eq:obsv:model} with gains $L_f$, $L_i$, and $L_o$, provides a converging state estimation, i.e. $\hat{\chi}(k) \to \chi(k)$ as $k\to \infty$, if $\rho(\hat{A}) < 1$, with
\begin{equation}\label{eq:obsv:A_hat}
	\hat{A} = \hat{A}(L_f, L_i, L_o) = \begin{bmatrix}
		\hat{\bar{\sigma}}_g^f & \hat{\alpha} \\
		\hat{\bar{\sigma}}_g^f \hat{\bar{\sigma}}_g^o &  \quad\frac{1}{4}  \bar{\sigma}_c^x \| U_o - L_o C \| + \hat{\bar{\sigma}}_g^o \hat{\alpha}
	\end{bmatrix},
\end{equation}
where
\begin{align} \label{eq:obsv:sigma_g_f}
	\sigma_g\left[W_f u + U_f \hat{\xi} + b_f + L_f(y - \hat{y}) \right] \leq \hat{\bar{\sigma}}_g^f &= \sigma_g( \| [ W_f u_{\scriptscriptstyle \text{max}} \quad U_f \! - \! L_f C \quad b_f \quad L_f C ] \|_\infty ), \\
	\sigma_g\left[W_i u + U_i \hat{\xi} + b_i + L_i(y - \hat{y}) \right] \leq \hat{\bar{\sigma}}_g^i &= \sigma_g( \| [ W_i u_{\scriptscriptstyle \text{max}} \quad U_i \! - \! L_i C \quad b_i \quad L_i C ] \|_\infty ), \\
	\sigma_g\left[W_o u + U_o \hat{\xi} + b_o + L_o(y - \hat{y}) \right] \leq \hat{\bar{\sigma}}_g^o &= \sigma_g( \| [ W_o u_{\scriptscriptstyle \text{max}} \quad U_o \! - \! L_o C \quad b_o \quad L_o C ] \|_\infty ), \label{eq:obsv:sigma_g_o} \\
 	\hat{\alpha} &= \frac{1}{4}\|U_f - L_f C\|   \frac{ \bar{\sigma}_g^i\bar{\sigma}_c^c}{1-\bar{\sigma}_g^f}+   \bar{\sigma}_g^i   \|U_c\|   + \frac{1}{4} \|U_i - L_i C\|    \bar{\sigma}_c^c. \label{eq:obsv:alpha_hat}
\end{align}
\end{theorem}

Note that the bounds \eqref{eq:obsv:sigma_g_f}-\eqref{eq:obsv:sigma_g_o} could be slightly more conservative than \eqref{eq:model:bound:sigma_g_f}-\eqref{eq:model:bound:sigma_g_o}, due to the presence of extra terms depending on $L_f$, $L_i$ and $L_o$.

\begin{proposition} \label{prop:obsv:tuning}
	A suitable tuning of state observer \eqref{eq:obsv:model}, guaranteeing the convergence of the state estimate, can be found solving the nonlinear optimization problem
	\begin{subequations}\label{eq:obsv:tuning}
		\begin{align}
		L_f^{*}, L_i^{*}, L_o^{*} = \arg\min_{L_f, L_i, L_o} & \| \hat{A} \| \\
		\text{s.t.} \quad & -1+\hat{\bar{\sigma}}_g^f +\alpha\hat{\bar{\sigma}}_g^o + \frac{1}{4}\bar{\sigma}_c^x\|U_o - L_o C\|  <\frac{1}{4}\hat{\bar{\sigma}}_g^f\bar{\sigma}_c^x\|U_o - L_o C \| < 1 \label{eq:obsv:tuning:constraint}
		\end{align}
	\end{subequations}
\end{proposition}

It should be noted that \eqref{eq:obsv:tuning} always admits a feasible solution, corresponding to $L_f=L_i=L_o=0_{n_x,n_y}$. Indeed, in such case  (corresponding to an open-loop state observer) $\hat{A} = A_\delta$, which by assumption is stable. 
However, it would be advisable to employ the output measurements to have a more reliable state estimation and a faster convergence of the state prediction error to zero.
Recalling that $\rho(\hat{A}) \leq \| \hat{A} \|$, minimizing $\| \hat{A} \|$ allows to make the observer faster in the worst case, and likely faster in general, while its Schur stability is enforced via \eqref{eq:obsv:tuning:constraint}.

\subsection{Model Predictive Control design} \label{sec:mpc}
This section discusses the design of a predictive control scheme that takes advantage of the LSTM network \eqref{eq:model:lstm} as a prediction model of the system. The objective of the controller is to stabilize the system towards a generic equilibrium point denoted by the triplet $(\bar{u},\bar{\chi},\bar{y})$ (where $\bar{u} \in \mathcal{U}$ and $\bar{\chi}=\begin{bmatrix}\bar{x}^T&\bar{\xi}^T\end{bmatrix}^T$) by suitably acting on the control input $u$. Let us define
\begin{equation}\label{eq:delta_X}
\Delta=\begin{bmatrix}
\|x-\bar{x}\|\\
\|\xi-\bar{\xi}\|
\end{bmatrix}\in \mathbb{R}^2.
\end{equation}
The MPC scheme consists of solving, at each sampling time $k$, the following optimization problem 
\begin{equation}\label{eq:mpc:fist_mpc}
\begin{aligned}
	\min_{U(k)} \quad &J(U(k)), \\
	\text{s.t.} \quad & u(k+i)\in \mathcal{U} \,\,\, \text{ for } i\in \{0,\dots,N-1\}, \\
	& \rev{\chi(k+i) = f(\chi(k), u(k)) \,\,\, \text{ for } i\in \{1,\dots,N\}}, 
\end{aligned}
\end{equation}
\rev{where $f(\chi, u)$ is defined by the model dynamics \eqref{eq:model:lstm_x} and \eqref{eq:model:lstm_xi}, and $U(k)= \begin{bmatrix}u(k)&\dots& u(k+N-1)\end{bmatrix}^T$ is the sequence of future control moves, which are bounded as in \eqref{eq:model:bound_u}, consistently with Assumption \ref{ass:u_bounded}}.
The terms $\chi(k+i)$, with $i \in \{1,\dots, N\}$, are the future state predictions from the current time instant to the prediction horizon $N$. These terms are obtained by iterating \eqref{eq:model:lstm} starting from the current state estimated by the state observer \eqref{eq:obsv:model}, i.e. $\chi(k)=\hat{\chi}(k)$.
The cost function reads as
\begin{equation} \label{eq:mpc:costfunction} J(U(k))\!=\!\!\sum\limits_{i=0}^{N-1}\biggl(\|\chi(k+i)\!-\!\bar{\chi}\|_Q^2\!+\!\|u(k+i)\!-\!\bar{u}\|_R^2\biggr)\!+\! \|\Delta(k+N)\|_P^2.
\end{equation}
Matrices $Q \succ 0$ and $R \succ 0$ are tuning parameters, and matrix $P\succ 0 \in \mathbb{R}^{2\times 2}$ satisfies the Lyapunov condition - under the assumption that $\rho(A_\delta)<1$:
\begin{equation} \label{eq:mpc:P_matrix}
A_\delta^TPA_\delta - P + qI_2 \prec 0,
\end{equation} where $q=\rho(Q)$.
At time step $k$, the solution to the optimization problem is termed $U(k|k)= \begin{bmatrix}u(k|k)^T&\dots& u(k+N-1|k)^T\end{bmatrix}^T$.
Only its first element is applied to the system according to the Receding Horizon principle, i.e.,
\begin{equation} \label{eq:mpc:receiding}
u(k)=u(k|k)
\end{equation}
The following result holds, ensuring asymptotic stability of the equilibrium $(\bar{u}, \bar{\chi}, \bar{y})$ under the proposed control law \eqref{eq:mpc:receiding}.
\begin{theorem}\label{thm:mpc_stabilizing}
If the plant behaves according to \eqref{eq:model:lstm} with $\rho(A_\delta)<1$, the state observer \eqref{eq:obsv:model} is designed so as to fulfill Theorem \ref{thm:obsv_convergence}, e.g. by means of Proposition \ref{prop:obsv:tuning}, then $(\bar{u},\bar{\chi},\bar{y})$ is an asymptotically stable equilibrium under the control law \eqref{eq:mpc:receiding}.
\end{theorem}
\begin{proof}
See the Appendix.
\end{proof}

%\begin{remark} \label{remark:mpc_nominal}
%Theorem \ref{thm:mpc_stabilizing} relies on the \rev{Certainty Equivalence Principe}, i.e. on the assumption that the plant is exactly described by the trained LSTM model \eqref{eq:model:lstm}.
%Such assumption is rather standard in MPC results, \rev{especially when black-box identification procedures are adopted to retrieve a model of the system}.
%\hl{Non so se inserirei la parte successiva ->} Nonetheless, this assumption is here supported by the fact that, if the real system is stable and suitably trained, a NN has been proved to reproduce any plant dynamics (i.e. any function $\varphi(\cdot,\cdot)$, see \eqref{eq:intro:discretetime} ) with arbitrary precision \cite{chen1995universal,siegelmann1992computational}.
%\end{remark}

\subsection{\rev{Model-plant mismatch and robustness}} \label{sec:scenario}
%textcolor{red}
The result reported in Theorem \ref{thm:mpc_stabilizing} relies on the assumption that the plant is exactly described by the LSTM model \eqref{eq:model:lstm} and there is no model-plant mismatch.
However, the LSTM model is obtained by means of a training (identification) procedure starting from input-output data samples collected on the real plant.
This raises the fundamental issues of reliability and robustness. % guaranteed by any control design method based on a learned black-box model.
The robust design of state-feedback MPC algorithms for nonlinear systems has been considered in a number of papers \cite{mayne2011tube,falugi2013getting,kohler2020computationally}, while the output feedback case has also been analysed for systems affected by a (state and/or input dependent) disturbance acting on the state equation \cite{kohler2019simple}.
However, to the best of the authors' knowledge, no results are available concerning the robust output-feedback MPC design for nonlinear black box estimated models, where also an estimation of the size of the disturbance must be computed.
For this reason, in this section we sketch a possible procedure towards the design of robust MPC for black-box models learned by LSTM networks. \\
Assume that the real system is described by model \eqref{eq:model:lstm} plus a disturbance $w$ acting on the output, representing the effect of the modeling error:
\begin{equation}\label{eq:robust:real_system}
%\nonumber to remove numbering (before each equation)
\begin{aligned}
	\chi^+ &= f(\chi,u), \\
	y_m &= g(\chi,u)+w ,
\end{aligned}
\end{equation}

\noindent The function $f(\chi,u)$ is defined by \eqref{eq:model:lstm_x} and \eqref{eq:model:lstm_xi}, $g(\chi,u)$ is defined by \eqref{eq:model:lstm_output}, $y_m$ is the measured output variable. \\
It follows that the output estimation error is
\begin{equation}\label{eq:robust:w}
w=y_m-y
\end{equation}
where $y$ is the output of the model \eqref{eq:model:lstm} fed by the same input sequence $\bm{u} = \{u(0), ..., u(\tau)\}$. 
Note that if Theorem \ref{thm:iss} is fulfilled and if the plant can be represented as \eqref{eq:robust:real_system}, both the LSTM model and the plant are ISS, and thus $w$ is guaranteed to be bounded.
In the following we describe an algorithm to estimate from the data, with a probabilistic accuracy guarantee, the smallest ball $\mathcal{W}$, with radius $\rho_w^*$, containing $w$.
Such algorithm relies on the Scenario Approach \cite{bonassi2019lstm, campi2009scenario, hewing2019scenario}. 
To this end, let the initial state of the LSTM network $\chi_0$ be a random variable extracted from a set $\mathcal{X}_0$, with some probability measure $\mathbb{P}_{\chi}$. 
Moreover, consider a class $\bm{\mathcal{U}}_{\tau}$ of input sequences $\{u(0), \ldots, u(\tau)\}$, $\tau$ being the adopted time horizon, such that $u(k) \in \mathcal{U}$ for all $k=0, \ldots, \tau$.
Assume that $\bm{\mathcal{U}}_\tau$ is characterized by some probability measure $\mathbb{P}_u$.
%, and that the class of inputs $\bm{\mathcal{U}}_\tau$ is compliant with the policy by which the system is operated.
Then, the radius $\rho_w^*$ is defined as the solution of
\begin{subequations} \label{eq:verification:inf}
	\begin{align}
	\rho_w^*=\min_{\rho_w} \quad & \rho_w, \\
	s.t. \quad &  \| \mathbf{w}(\chi_{0}, \mathbf{u}) \|\sss{\infty} \leq \rho_w  \qquad \forall{\chi_0}\in \mathcal{X}_0, \, \forall \mathbf{u} \in \bm{\mathcal{U}}\tau, \label{eq:verification:inf:constr}
	\end{align}
\end{subequations}
where $\mathbf{w}(\chi_{0}, \mathbf{u})$ is the disturbance sequence obtained feeding the trained LSTM \eqref{eq:model:lstm} with the input sequence $\mathbf{u}$ and initial state $\chi_{0}$. \\
Problem \eqref{eq:verification:inf} cannot be solved directly, due to infinite cardinality of constraint  \eqref{eq:verification:inf:constr}.
Nonetheless, owing to the convexity of \eqref{eq:verification:inf} with respect to the optimization variable $\rho_w$, the Scenario Approach can be exploited to recast the optimization problem as a finite-dimensional linear program, that allows to compute $\rho_w^*(\varepsilon,\beta)$ such that, with confidence $1 - \beta$,
\begin{equation}
\mathbb{P}_{\chi, u} \left\{  \| \mathbf{w}(\chi_{0}, \mathbf{u}) \|\sss{\infty} > \rho_w^*(\varepsilon,\beta) \right\} \leq \varepsilon.
\end{equation}
To do so, it is necessary to generate $K$ realizations of the uncertain variables $\chi_0$ and $\mathbf{u}$, denoted by $\chi_0^{(i)},\mathbf{u}^{(i)}$, drawn according to the respective probability density functions \cite{bonassi2019lstm}.
It has been shown \cite{campi2009scenario} that $\rho_w^*(\varepsilon,\beta)$ can be computed as
\begin{equation} \label{eq:verification:scenario}
\begin{aligned}
\rho_w^*(\varepsilon,\beta)=\min_{\rho_w} \quad & \rho_w, \\
s.t. \quad &  \| \mathbf{w}(\chi_{0}^{(i)}, \mathbf{u}^{(i)}) \|\sss{\infty} \leq \rho_w  \,\,\, \text{ for all } i=1, ..., K,
\end{aligned}
\end{equation}
provided that the number of scenarios satisfies the following inequality:
\begin{equation} \label{eq:verification:Nscen}
K \geq \frac{2}{\varepsilon} \left( \ln\frac{1}{\beta} + d \right).
\end{equation}
Once that the bound $\rho_w$ is known, it is possible to design a state observer for system \eqref{eq:robust:real_system}.
With mathematical development analogous to those of Theorem \ref{thm:obsv_convergence}, and under similar conditions, it can be proven that the state estimation error does not vanish, but asymptotically converges to an invariant set whose size depends on $\rho_w$ itself.
Based on this state estimate, robust state-feedback \cite{kohler2020computationally} or output-feedback \cite{kohler2019simple} control laws can be designed.

\section{Illustrative example} \label{sec:example}
The benchmark example here considered to test the described identification and control algorithm is a PH neutralization process \cite{hall1989modelling}, composed of two tanks, namely Tank 1 and Tank 2, see also Figure \ref{fig:Ph_scheme}.

\begin{figure}
	\centering
	\includegraphics[scale=0.3]{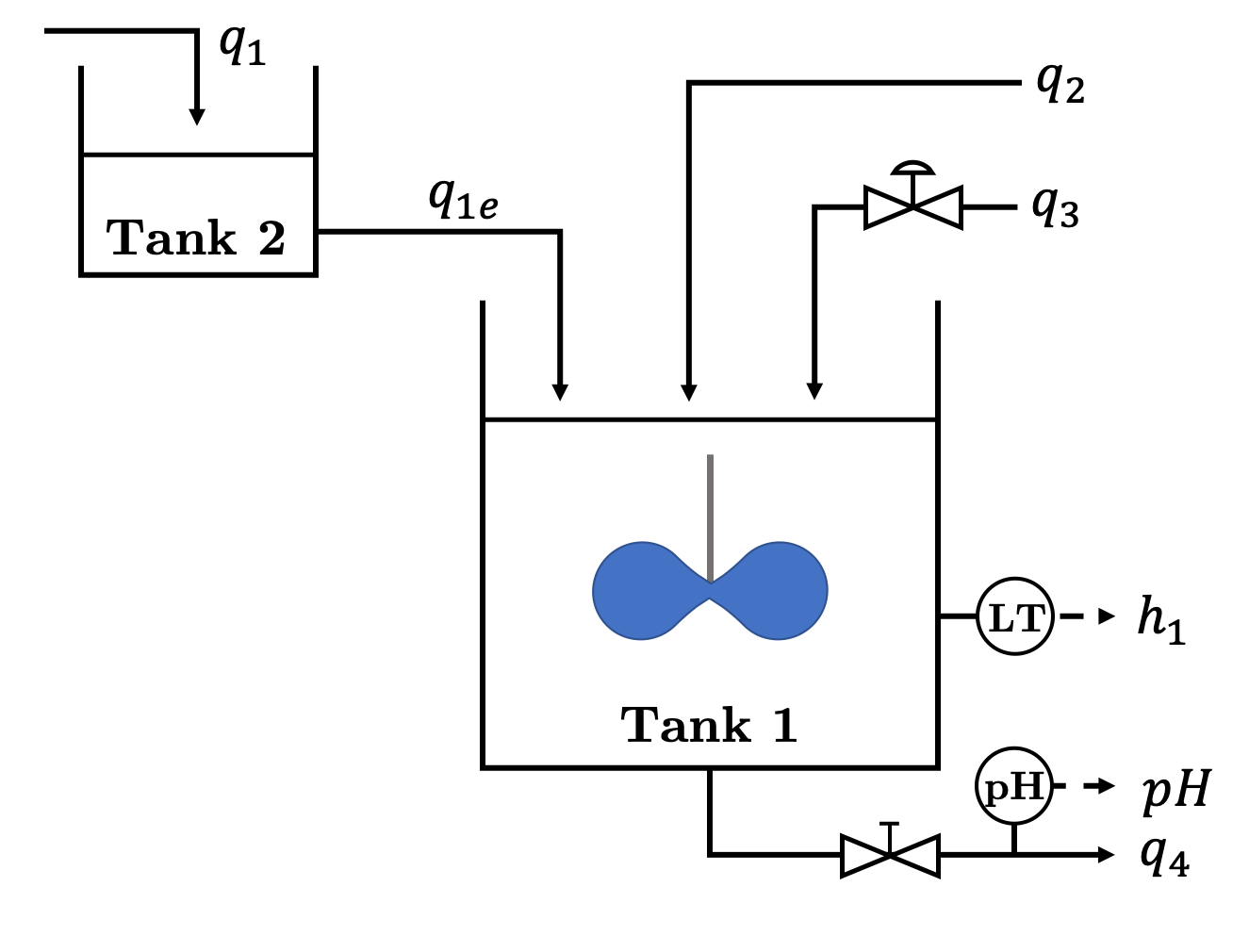}
	\caption{Schematic layout of the PH neutralization process}
	\label{fig:Ph_scheme}
\end{figure}

Tank 2 is fed by an acid stream $q_1$ and outputs a flow $q_{1e}$, but this hydraulic dynamics is neglected being much faster than the others involved, so that it can be assumed $q_{1}=q_{1e}$. Tank 1, also called reactor tank, is instead fed by three flows, namely $q_1$, a buffer flow $q_2$ and an alkaline flow $q_3$.
$q_1$ and $q_2$ are not manipulated variables, and represent disturbances, whereas a controlled valve modulates $q_3$. On the output side the flow is $q_4$, where the PH is measured. The objective of the control scheme is to stabilize the PH concentration to a desired value.
The plant is characterized by the following set of differential equations with a constraint \cite{hall1989modelling}:
\begin{equation}
\begin{aligned}
\dot{x}(t)= f_1(x(t))+f_2(x(t)) & u(t)+f_3(x(t))d(t), \\[1.5ex]
c(x(t),y(t)) & = 0,
\end{aligned}
\label{eq:processEq1}
\end{equation}
where
\begin{equation*}
	\begin{aligned}
	f_1(x(t)) & = \left[\frac{q_{1}}{A_{1}x_{3}}(W_{a1}-x_{1}),\frac{q_{1}}{A_{1}x_{3}}(W_{b1}-x_{2}),\frac{1}{A_{1}}(q_{1}-C_{v4}(x_{3}+z)^{n})\right]^{T}, \\
	f_2(x(t)) & = \left[\frac{1}{A_{1}x_{3}}(W_{a3}-x_{1}),\frac{1}{A_{1}x_{3}}(W_{b3}-x_{2}),\frac{1}{A_{1}}\right]^{T}, \\
	f_3(x(t)) & = \left[\frac{1}{A_{1}x_{3}}(W_{a2}-x_{1}),\frac{1}{A_{1}x_{3}}(W_{b2}-x_{2}),\frac{1}{A_{1}}\right]^{T}, \\
	c(x,y) & = x_{1}+10^{y-14}+10^{-y}+x_{2}\frac{1+2\cdot 10^{y-pK_{2}}}{1+10^{pK_{1}-y}+10^{y-pK_{2}}}. 
	\end{aligned}
\end{equation*}

\vspace{2mm}
\noindent The parameters $pK_{1}$ and $pK_{2}$ are the first and second dissociation constants of the weak acid $H_{2}CO_{3}$.
The nominal values of the model parameters are given in Table \ref{tab:Tab3}, where $[M]=[\frac{mol}{L}]$. Overall, the simplified model considered is of order three, with one input and one output.

\footnotesize
\begin{table}[h]
	\centering	
	\caption{Nominal operating conditions of the $pH$ system}
	\begin{tabular}{lll}
		\toprule
		$z = 11.5\,cm$ 			& $W_{a1} = 3.00\cdot10^{-3}M$ &$q_{1} = 16.6\,mL/s$  \\
		$C_{v4} = 4.59$ 	& $W_{b1} = 0.00M$ &$q_{2} = 0.55\,mL/s$  \\
		$n = 0.607$ 			& $W_{a2} = -0.03M$ &	$q_{3} = 15.6\,mL/s$ \\
		$pK_{1} = 6.35$ 		& $W_{b2} = 0.03M$ &$q_{4} = 32.8\,mL/s$ \\
		$pK_{2} = 10.25$ 		& $W_{a3} = 3.05\cdot10^{-3}M$ &$A_{1} = 207\,cm^{2}$  \\
		$h_{1} = 14\,cm$		& $W_{b3} = 5.00\cdot10^{-5}M$ & $W_{a4} = -4.32\cdot10^{-4}M$\\
		$pH = 7.0$& $W_{b4} = 5.28\cdot10^{-4}M$\\
		\bottomrule
	\end{tabular}
	\label{tab:Tab3}
\end{table}
\normalsize

\subsection{Identification}
The simulator of the plant has been forced with a multilevel pseudo-random signal (MPRS), so as to properly excite the system, and the input-output data have been recorded with a sampling time $T_s=10 s$ so as to collect about 30-40 samples in a step response.
Also, to test the algorithm in a more realistic scenario, a white noise was added both to the input and output variables, with power $4 \times 10^{-4} W$.
\rev{The generated dataset consists of $15$ experiments for the training set, $4$ for the validation set, and $1$ for the test set, where each experiment is a collection of $N_s = 2000$  $\{ u(t), y(t)\}$ samples.The datasets have been normalized according to standard techniques \cite{goodfellow2016deep}, so that $u \in [-1, 1]$ and $y_m \in [-1, 1]$. }

\begin{figure}
	\centering
	\includegraphics[scale=0.7]{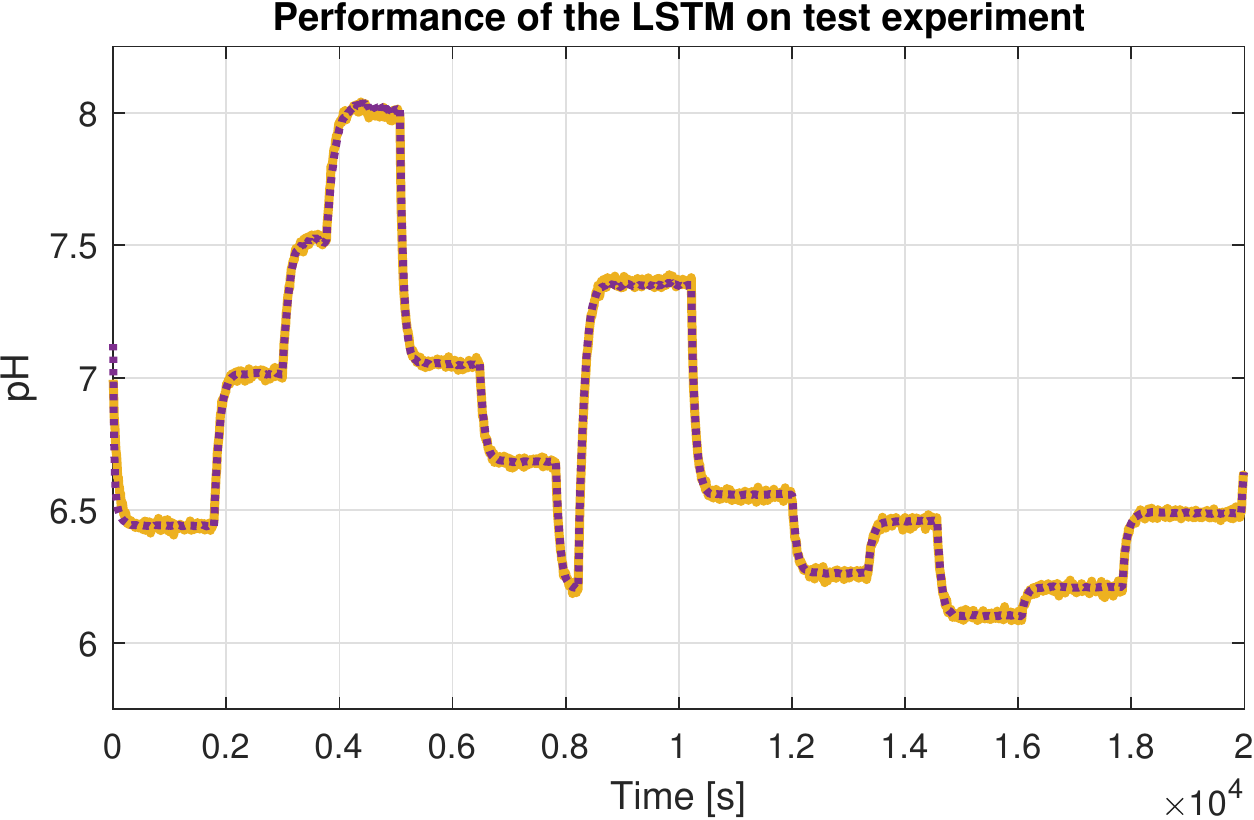}
	\caption{Performances of the trained LSTM on the independent test dataset: LSTM prediction (purple dotted line) compared to the real measured output (yellow line).}
	\label{fig:validation}
\end{figure}

\rev{The LSTM network described by \eqref{eq:model:lstm} with $n_x=7$ neurons has been implemented and trained in Python 3.7 using Tensorflow 1.15. 
The training procedure has been carried out with the RMSProp algorithm \cite{goodfellow2016deep}, using randomly picked single experiments as mini-batches, to minimize the following loss function
\begin{subequations}
\begin{align}
	L =& \frac{1}{N_s} \sum_{k=0}^{N_s-1} \| y_m(k) - y(k) \|_2^2 + \label{eq:example:loss:mse} \\
	& -\rho_1^{+} \max (r_1, 0) - \rho_1^{-} \min (r_1, 0) -\rho_2^{+} \max (r_2, 0) - \rho_2^{-} \min (r_2, 0), \label{eq:example:loss:residuals}
\end{align} \label{eq:example:loss_tot}
\end{subequations}
where \eqref{eq:example:loss:mse} is the MSE between $y_m$, i.e. the experiment's measured output, and $y$, i.e. the output predicted by the LSTM network.
Furthermore, as hinted in Remark \ref{remark:train}, in order to obtain a network enjoying the $\delta$ISS property -- and, in light of Corollary \ref{prop:deltaiss_implies_iss}, the ISS as well -- the residuals of the inequalities \eqref{eq:deltaiss:condition}, have been placed in the loss function, see \eqref{eq:example:loss:residuals}. These residuals, denoted by $r_1$ and $r_2$, are defined as 
\begin{subequations}
	\begin{align}
		r_1 &= -1+ \bar{\sigma}_g^f +\alpha\bar{\sigma}_g^o + \frac{1}{4}\bar{\sigma}_c^x\|U_o\| - \frac{1}{4}\bar{\sigma}_g^f\bar{\sigma}_c^x\|U_o\|, \\
		r_2 &= \frac{1}{4}\bar{\sigma}_g^f\bar{\sigma}_c^x\|U_o\| - 1.
	\end{align}
\end{subequations}
Note that in \eqref{eq:example:loss:residuals} a piece-wise linear reward has been adopted so as to avoid an uselessly large fulfillment of \eqref{eq:deltaiss:condition} at expenses of the fitting quality. 
Indeed, the weights $\rho_{1,2}^\pm$ have been chosen sufficiently small to ensure that the MSE dominates the loss function ($\rho_1^+ = \rho_2^+  = 4 \cdot 10^{-3}$, $\rho_1^- = \rho_2^- = 2 \cdot 10^{-5}$). An early-stopping rule has also been implemented to stop the training procedure when for a pre-defined number of epochs the MSE over the validation set does not improve.}

\rev{The trained network, as well as the datasets, is publicly available \cite{enrico_terzi_2020_3956067}.} The modeling performances over an independent test set are reported in Figure \ref{fig:validation}, where the measured output is compared to the prediction of the trained network, initialized from a random value, and forced by the same input $u$.
A quantitative performance index is the FIT $[\%]$ value, which is computed as
\begin{equation} \label{eq:example:fit}
\text{FIT}=100\left(1- \frac{\|\bm{y}_m - \bm{y}\|}{\|\bm{y}_m - y_{avg}\|} \right)
\end{equation}
where $\bm{y}_m$ collects the output samples of the dataset, $y_{avg}$ is its average and $\bm{y}$ collects the output simulation of the trained LSTM.
Over the independent test dataset, the FIT scores $96.5\%$, mainly due to noise, thus confirming remarkable modeling properties.
\rev{The designed cost function \eqref{eq:example:loss_tot} ensured the satisfaction of the $\delta$ISS condition \eqref{eq:deltaiss:condition}, with residuals $r_1=-0.049$ and $r_2= -0.997$.
Indeed, the parameters of the trained LSTM lead to $\bar{\sigma}_g^f = 0.82$, $\bar{\sigma}_g^i = 0.51$, $\bar{\sigma}_c^c = 0.93$, $\bar{\sigma}_g^o = 0.61$, $\bar{\sigma}_c^x = 0.99$, $\|  U_f \| =  0.01$, $\|  U_i \| =  0.01$, $\|  U_c \| =  0.4$, $\|  U_o \| =  0.01$, therefore $A_\delta = \begin{bmatrix}
0.819  &  0.215 \\
0.502  &  0.135
\end{bmatrix}$ and $\rho(A_\delta) = 0.95$, i.e. $A_\delta$ is Schur stable.}

\subsection{Control}
The designed observer follows \eqref{eq:obsv:model} and is tuned according to \eqref{eq:obsv:tuning}, thus guaranteeing a realiable state estimate to the MPC controller.
The testing experiment is a reference tracking one. More specifically, the controller is started at time $500 s$, and it is required to track a setpoint reference $\bar{y} \in \{7, 8, 6.5, 6, 7.5\}$, and to stabilize the associated equilibrium $(\bar{u}, \bar{\chi}, \bar{y})$ \rev{of the LSTM model. Therefore, $\bar{u}$ and $\bar{\chi}$ have been numerically computed from \eqref{eq:model:lstm} by setting $u(k) = \bar{u}$, $\chi(k+1) = \chi(k) = \bar{\chi}$ and $y(k) = \bar{y}$}. 
% Exploiting the designed observer, and assuming to know the pairs $(\bar{u},\bar{y})$, the corresponding steady state $\bar{\chi}$ is obtained by simulation, in view of the guaranteed convergence properties of the observer itself.

\rev{The adopted prediction horizon is $N = 10$ steps}, and matrices in the cost function of the controller are $Q=I_{2n_x}$, $R=5$, while the terminal weight matrix $P$ is computed according to \eqref{eq:mpc:P_matrix}.

The closed-loop trajectory is reported in Figure \ref{fig:closed_loop}, which shows that the controller is able to effectively manage the plant, fulfilling control constraints and improving the transient responses. In particular, note that around 2000 $s$ the input is saturated to its upper bound. 
To confirm the validity of the estimate provided by the observer, in Figure \ref{fig:output_estimate}, the output estimate $\hat{y}$ and the real plant output are compared, showing the convergence of the estimate, save for a static mismatch due to the model (LSTM) - plant (pH simulator) gain mismatch when $pH \simeq 8$.

\begin{figure}
	\centering
	\subfloat[]{\includegraphics[width=0.475 \linewidth]{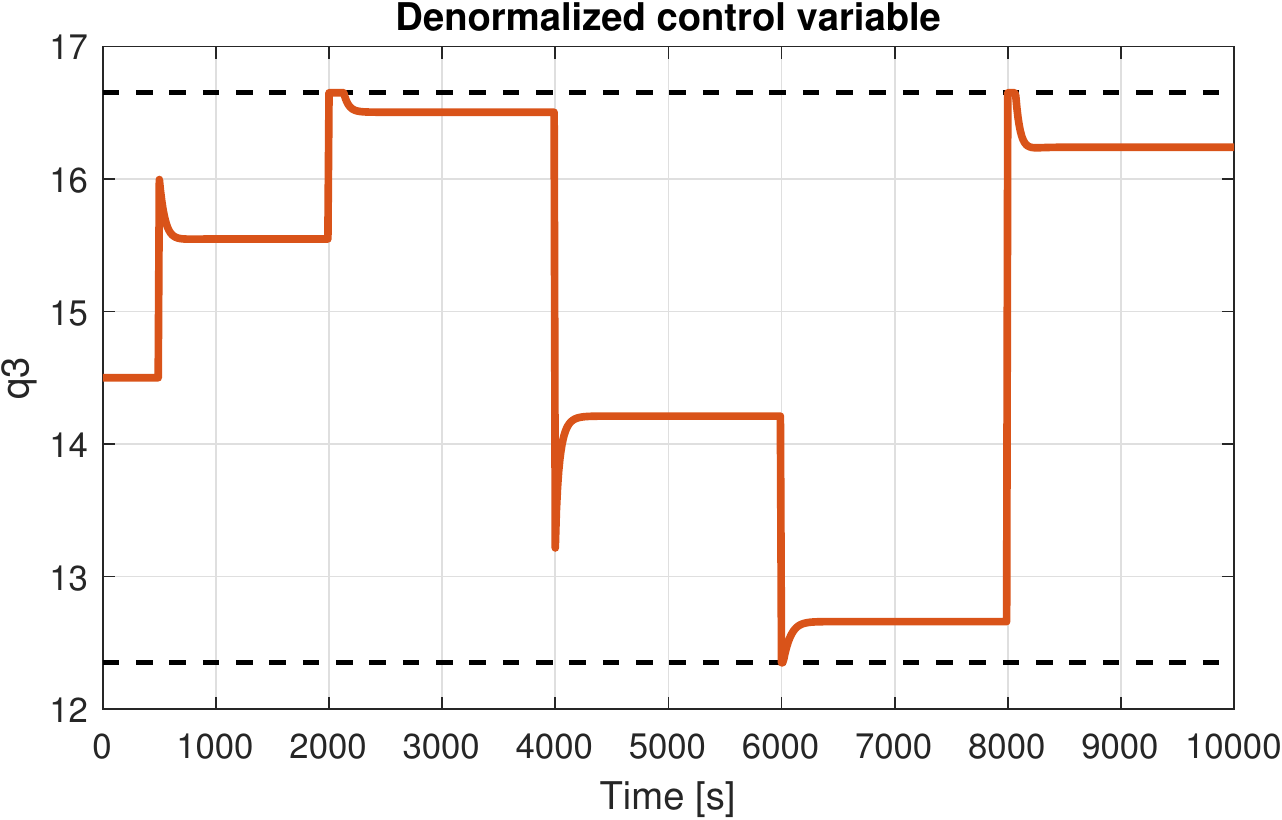}} \quad
	\subfloat[]{\includegraphics[width=0.475 \linewidth]{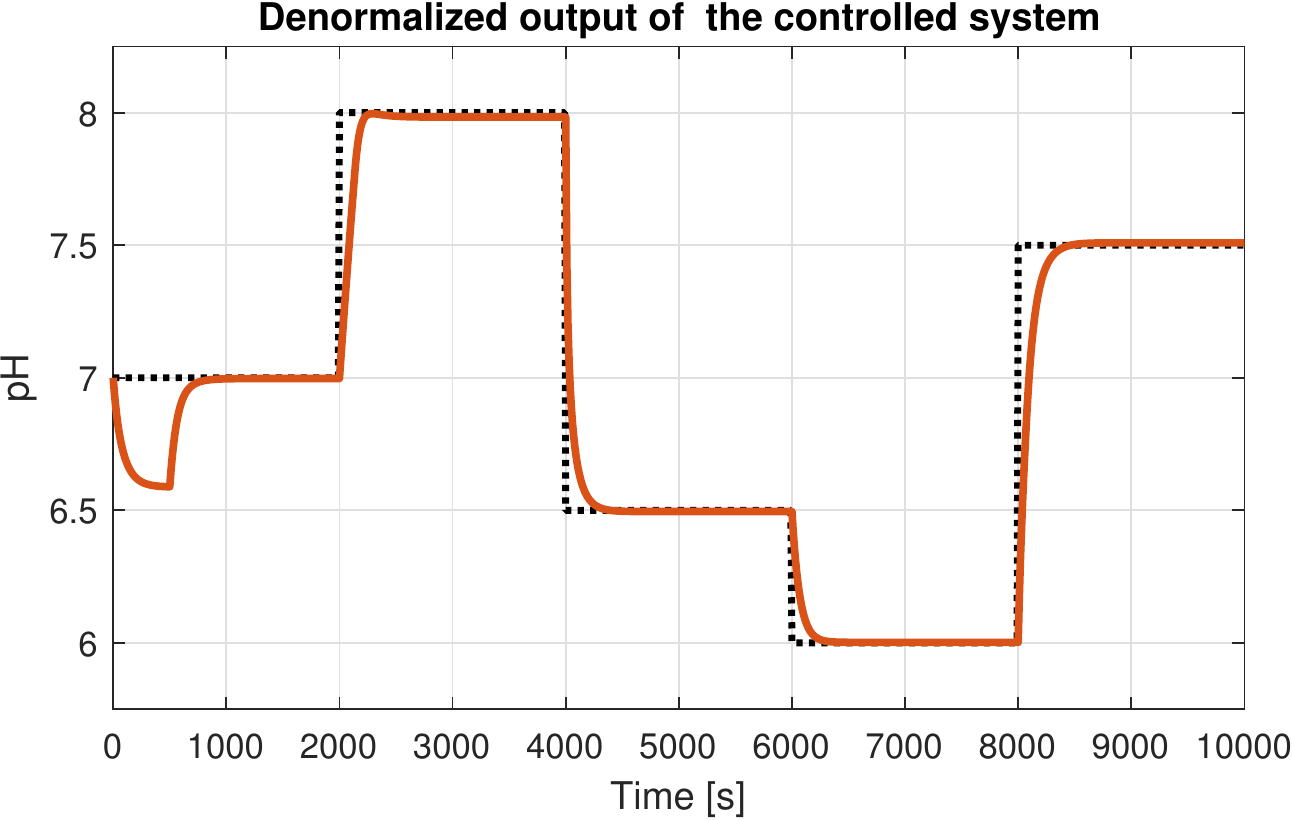}}
	\caption{Closed loop trajectories. (a) Denormalized input trajectory (red line) compared to the lower and upper bounds (dashed lines); (b) Denormalized closed-loop output trajectory (red line), compared to the output reference (dotted line).}
	\label{fig:closed_loop}
\end{figure}
\begin{figure}
	\centering
	\includegraphics[width=0.475 \linewidth]{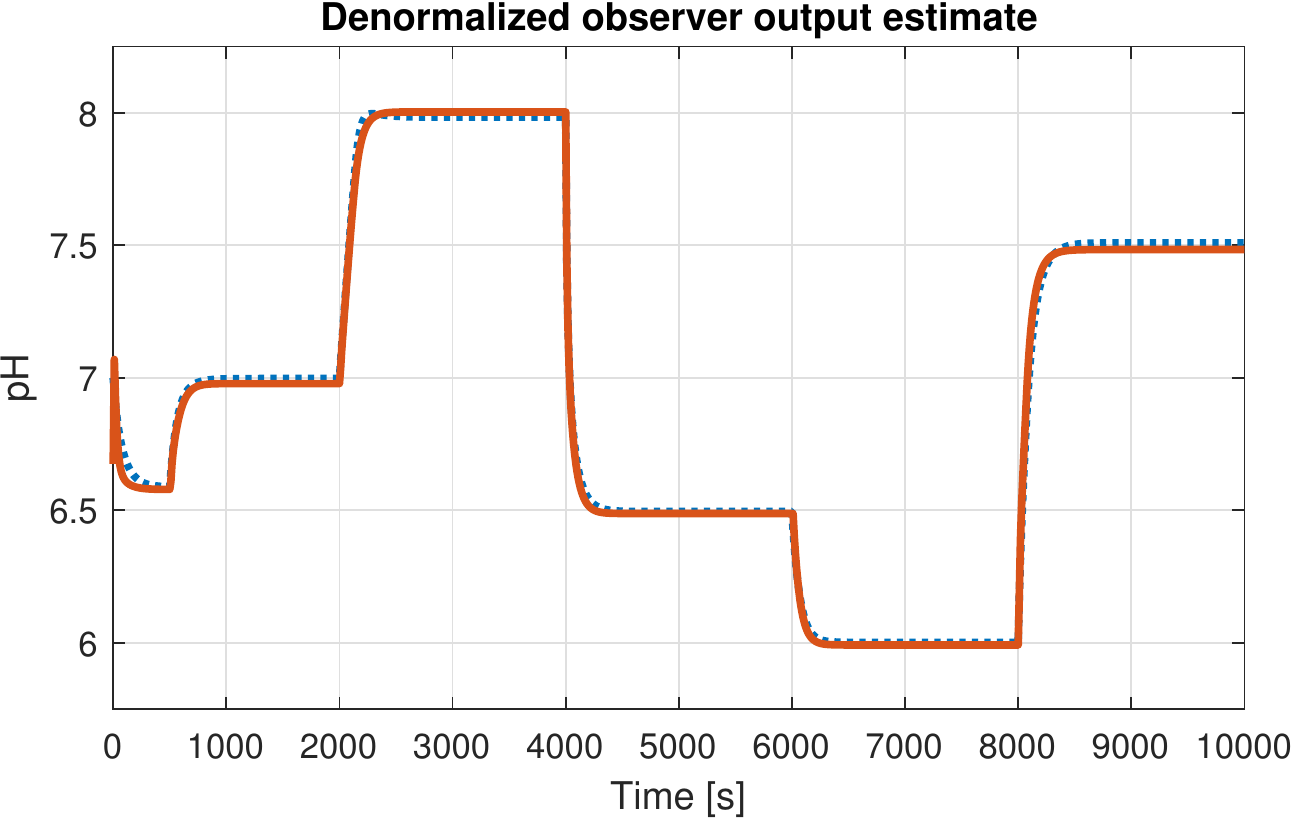}
	\caption{Denormalized output estimated by the state observer \eqref{eq:obsv:model} (red continuous line) compared to the real plant output (blue dotted line).}
	\label{fig:output_estimate}
\end{figure}

\section{Conclusion} \label{sec:conclusions}
In this paper Long Short Term Memory networks have been investigated from a system theoretical perspective, and sufficient conditions for their ISS and $\delta$ISS stability properties have been provided in terms of their internal weights. A novel formulation of the optimization problem to train the NN, including constraints, has been employed. The obtained NN has been then used as a prediction model in a MPC scheme endowed with an observer to suitably provide the initial state estimate, with guaranteed convergence of the estimate and asymptotic stability of the closed-loop equilibrium. Numerical results on a nonlinear SISO benchmark confirm the theoretical findings in the case of a tracking problem.

%Future work is concerned with the extension of the analysis to other classes of LSTM, e.g. the \textit{peephole} one including feedback of $x$ in the state equations, and with the development of robust control algorithms explicitly taking into account the model-plant mismatch.

\rev{Future work will be devoted to enhancing the robustness of the control algorithm with respect to model-plant mismatch. 
To this regard, the preliminary analysis developed in Section \ref{sec:scenario} for the estimation of a bound on the mismatch can be extended, and other approaches can be studied, like the one presented by Fazlyab et al. \cite{fazlyab2019efficient}, where however only feed-forward networks have been considered.}

\section{Appendix}
The following properties will be used in the proofs:

\begin{property} \label{prop:matrix_product_Hadamard}
Given vectors $v_1,v_2 \in \mathbb{R}^n$, $v_1 \circ v_2 = diag(v_1)v_2$.
\end{property}

\begin{property} \label{prop:A_spectral_radius}
Given a diagonal matrix $A$, $\|A\|=\rho(A)$, and the eigenvalues of $A$ are its diagonal entries.
\end{property}

\begin{property} \label{prop:norms}
	Given two vectors $a,b\in\mathbb{R}^{n}$ and a positive definite matrix $M \succ 0$, it holds that, for any $\nu \neq 0$, $$\|a+b\|_M^{2}\leq(1+\nu^{2})\|a\|_M^{2}+(1+\frac{1}{\nu^{2}})\|b\|_M^{2}.$$
\end{property}

First, let us introduce an instrumental Lemma which will be required in the following theoretical contribution.
\begin{lemma} \label{lemma:schur}
	Given a $2 \times 2$ real matrix $A$, it is Schur stable if and only if
	\begin{equation}
	-1-a < b < 1,
	\end{equation}
	where, being $A_{ij}$ the element of $A$ in position $(i,j)$, $a = - A_{11} - A_{22}$ and $b = {A_{11} A_{22} - A_{12} A_{21}}$.
\end{lemma}
\begin{proof}[Proof of Lemma \ref{lemma:schur}]
To characterize the stability of a $2 \times 2$ matrix $A$, let us compute its characteristic equation
\begin{align}
p(\lambda)=\text{det}(\lambda I_2-A)=\lambda^2 + a\lambda +b=0
\end{align}
where $a= - A_{11} - A_{22}$ and $b=A_{11} A_{22} - A_{12} A_{21}$.
We rely on Jury's criterion \cite{jury1962simplified}, providing a necessary and sufficient condition, to enforce stability of A. The Jury table of $p(\lambda)$ is
\begin{equation}
\begin{matrix}
1 & a & b\\
1-b^2 & a(1-b)& \\
\frac{(1-b^2)^2-a^2(1-b)^2}{(1-b^2)} & &
\end{matrix}
\end{equation}
Jury's criterion requires to force the first column to have all positive entries. This leads, with standard arguments and recalling that $a<0$, to the set of conditions:
\begin{equation}
\begin{cases}
b^2<1\\
1-b^2 > -a(1-b)
\end{cases}
\end{equation}
which can be further synthesized in
\begin{equation}
-1-a<b<1
\label{eq:a_b_condition}
\end{equation}
\end{proof}

\rev{\begin{proof}[Proof of Theorem \ref{thm:iss}]
	Let us consider the first LSTM state equation, i.e. \eqref{eq:model:lstm_x}. Taking the norm we get the following inequality
	\begin{equation*}
	\begin{aligned}
		\| x^+ \| &\leq \left\| \sigma_{g}\left( W_f u + U_f \xi + b_f \right) \right\| \| x \| + \left\| \sigma_{g}\left( W_i u + U_i \xi + b_i \right) \right\| \left\| \sigma_{c}\left( W_c u + U_c \xi + b_c \right) \right\|.
	\end{aligned}
	\end{equation*}
	In view of Properties \ref{prop:matrix_product_Hadamard} and \ref{prop:A_spectral_radius}, and owing to the Lipschitzianity of $\sigma_{c}$ and $\sigma_{g}$, recalling \eqref{eq:model:bound:sigma_g_f}-\eqref{eq:model:bound:sigma_c_x}, it holds that
	\begin{equation} \label{eq:proof:iss:xplus}
	\begin{aligned}
	\| x^+ \| &\leq \bar{\sigma}_g^f \| x \| + \bar{\sigma}_g^i \left[ \| W_c \| \| u \| + \| U_c \| \| \xi \| + \| b_c \| \right].
	\end{aligned}
	\end{equation}
	Then, with similar arguments, 
	\begin{equation} \label{eq:proof:iss:xiplus}
	\begin{aligned}
		\| \xi^+ \| &\leq \left\| \sigma_{g}\left( W_o u + U_o \xi + b_o \right) \right\| \| \sigma_{c}(x^+) \| \leq \bar{\sigma}_g^o \| x^+ \| \leq \bar{\sigma}_g^o \bar{\sigma}_g^f \| x \| + \bar{\sigma}_g^o \bar{\sigma}_g^i \| U_c \| \| \xi \| + \bar{\sigma}_g^o \bar{\sigma}_g^i \left[ \| W_c \| \| u \| + \| b_c \|\right].
	\end{aligned}
	\end{equation}
	Grouping  \eqref{eq:proof:iss:xplus} and \eqref{eq:proof:iss:xplus} we get
	\begin{equation} \label{eq:proof:iss:sys}
		\begin{bmatrix}
			\| x^+ \| \\
			\| \xi^+ \|
		\end{bmatrix} \leq A \begin{bmatrix}
			\| x \| \\
			\| \xi \|
		\end{bmatrix} + B_u \| u \| + B_b \| b_c \|,
	\end{equation}
	where $A = \begin{bmatrix}
			\bar{\sigma}_g^f & \bar{\sigma}_g^i \| U_c \| \\
			\bar{\sigma}_g^o \bar{\sigma}_g^f  & \bar{\sigma}_g^o \bar{\sigma}_g^i \| U_c \|
	\end{bmatrix}$ as in \eqref{eq:iss:A}, 
	$ B_u = \begin{bmatrix}
		\bar{\sigma}_g^i \| W_c \|  \\
		\bar{\sigma}_g^o \bar{\sigma}_g^i \| W_c \| 
	\end{bmatrix}$, and
	$ B_b = \begin{bmatrix}
		\bar{\sigma}_g^i \\
		\bar{\sigma}_g^o \bar{\sigma}_g^i
	\end{bmatrix}$.\\[2mm]
	Now we show that the stability of matrix $A$, i.e. $\rho(A) < 1$, entails the ISS property of the LSTM. By iterating \eqref{eq:proof:iss:sys} we get
	\begin{equation} \label{eq:proof:iss:int3}
			\begin{bmatrix}	\| x \| \\ \| \xi \| \end{bmatrix} 
			\leq A^k \begin{bmatrix}	\| x_0 \| \\ \| \xi_0 \| \end{bmatrix} + \sum_{i=0}^{k-1} A^{k-i-1}  \left( B_u \| u(i) \| + B_b \| b_c \| \right).
	\end{equation}
	Noting that 
	\begin{equation*}
	\left\| \begin{bmatrix}	\| x \| \\ \| \xi \| \end{bmatrix} \right\|
	= \left\| \begin{bmatrix} x  \\ \xi \end{bmatrix} \right\| = \| \chi \|,
	\end{equation*}
	taking the norm of \eqref{eq:proof:iss:int3} it follows that
	\begin{equation}
		\| \chi \| \leq \| A^k \| \| \chi_0 \| + \left\| \sum_{i=0}^{k-1} A^{k-i-1} \left( B_u \| u \| + B_b \| b_c \|\right) \right\|
	\end{equation}
	With standard norm arguments, since $A$ is Schur stable, there exist constants $\mu\geq 1$ and $\lambda\in(0,1)$ such that
	\begin{equation}
	\|\chi(k)\| \leq \mu\lambda ^k \|\chi_{0} \| + \|(I_{2}-A)^{-1}B_u\| \max_{h\geq 0}\|{u}(h)\| +  \|(I_{2}-A)^{-1}B_b \| \|b_c\| 
	\end{equation}
	i.e., that \eqref{eq:model:lstm} is ISS according to Definition~\ref{def:ISS}.
\end{proof}
\begin{proof}[Proof of Proposition \ref{prop:iss:schur}]
	Applying Lemma \ref{lemma:schur} to \eqref{eq:iss:A}, being  $a = -\bar{\sigma}_g^f -\bar{\sigma}_g^o \bar{\sigma}_g^i \| U_c \|$ and $b=0$, we conclude that a necessary and sufficient condition for the Schur stability of $A$ is that
	\begin{equation}
		-1 +  \bar{\sigma}_g^f +\bar{\sigma}_g^o \bar{\sigma}_g^i \| U_c \| <0 < 1,
	\end{equation}
	from which \eqref{eq:iss:condition} can be easily derived.
\end{proof}} 

\begin{proof}[Proof of Theorem \ref{thm:deltaISS}]
Let us compute the evolution of the upper bound of the norms of the two components of the state $\chi$, namely $x$ and $\xi$.
We proceed by addressing the two subvectors separately.
\begin{equation}
\begin{aligned}
	x_1^+-x_2^+ =& 	\, \sigma_g(W_f u_1 + U_f \xi_1 + b_f) \circ x_1 + \sigma_g(W_i u_1 + U_i \xi_1 + b_i) \circ \sigma_c(W_c u_1 + U_c \xi_1 + b_c)\\
		 	& -[\sigma_g(W_f u_2 + U_f \xi_2 + b_f) \circ x_2+\sigma_g(W_i u_2 + U_i \xi_2 + b_i) \circ \sigma_c(W_c u_2 + U_c \xi_2 + b_c)]\\
	=& \, \sigma_g(W_f u_1 + U_f \xi_1 + b_f)\circ (x_1-x_2)+x_2 \circ [\sigma_g(W_f u_1 + U_f \xi_1 + b_f)-\sigma_g(W_f u_2 + U_f \xi_2 + b_f)] \\
			&+\sigma_g(W_i u_1 + U_i \xi_1 + b_i) \circ [\sigma_c(W_c u_1 + U_c \xi_1 + b_c)-\sigma_c(W_c u_2 + U_c \xi_2 + b_c)] \\
			&+\sigma_c(W_c u_2 + U_c \xi_2 + b_c)\circ [\sigma_g(W_i u_1 + U_i \xi_1 + b_i) -\sigma_g(W_i u_2 + U_i \xi_2 + b_i)]
\end{aligned}
\end{equation}
Recalling the upper bounds \eqref{eq:model:bound:sigma_g_f}-\eqref{eq:model:bound:sigma_c_c}, Lipschitzianity of $\sigma_c(\cdot)$ and $\sigma_g(\cdot)$ and taking the norms both sides, we write, in view of Properties \ref{prop:matrix_product_Hadamard} and \ref{prop:A_spectral_radius}, it follows that
\begin{equation}\label{eq:proof:deltaiss:x1_x2}
\begin{aligned}
	\|x_1^+-x_2^+\| \leq & \, \bar{\sigma}_g^f \|x_1-x_2\| + \frac{\bar{\sigma}_g^i\bar{\sigma}_c^c}{1-\bar{\sigma}_g^f} \frac{1}{4} \bigg[ \|W_f\| \|u_1-u_2\| + 
			\|U_f\| \|\xi_1-\xi_2\| \bigg] + \\
			& + \bar{\sigma}_g^i\bigg(\|W_c\| \|u_1-u_2\| + \|U_c\| \|\xi_1-\xi_2\| \bigg)+ \bar{\sigma}_c^c\frac{1}{4}\bigg(\|W_i\| \|u_1-u_2\| + \|U_i\| \|\xi_1-\xi_2\| \bigg) \\
	\leq 	& \, \bar{\sigma}_g^f \|x_1-x_2\| + \alpha \|\xi_1-\xi_2\| + \beta\|u_1-u_2\|,
	\end{aligned}
\end{equation}
where $ \alpha= \bigg[\frac{1}{4}\|U_f\|   \frac{\bar{\sigma}_g^i\bar{\sigma}_c^c}{1-\bar{\sigma}_g^f}+   \bar{\sigma}_g^i   \|U_c\|   + \frac{1}{4} \|U_i\|    \bar{\sigma}_c^c       \bigg]$ and $\beta=\bigg[\frac{1}{4}\|W_f\|   \frac{\bar{\sigma}_g^i\bar{\sigma}_c^c}{1-\bar{\sigma}_g^f}+   \bar{\sigma}_g^i   \|W_c\|   + \frac{1}{4} \|W_i\|    \bar{\sigma}_c^c\bigg]$. \\[2mm]
Concerning the second state sub-vector,
\begin{equation}
	\begin{aligned}
	\xi_1^+-\xi_2^+ =& \, \sigma_{g}(W_o u_1 + U_o \xi_1 + b_o) \circ \sigma_c(x_1^+) - \sigma_{g}(W_o u_2 + U_o \xi_2 + b_o)\circ \sigma_c(x_2^+) \\
	=& \, \sigma_{g}(W_o u_1 + U_o \xi_1 + b_o)\circ ( \sigma_c(x_1^+)-\sigma_c(x_2^+) ) +\sigma_c(x_2^+)\circ\bigg[\sigma_{g}(W_o u_1 + U_o \xi_1 + b_o) - \sigma_{g}(W_o u_2 + U_o \xi_2 + b_o)\bigg].
	\end{aligned}
\end{equation}
By recalling the bounds \eqref{eq:model:bound:sigma_g_f}-\eqref{eq:model:bound:sigma_c_c}, \eqref{eq:model:bound:sigma_c_x} and \eqref{eq:proof:deltaiss:x1_x2}, taking the norm both sides, we get
\begin{equation}\label{eq:proof:deltaiss:xi1_xi2}
\begin{aligned}
\|\xi_1^+-\xi_2^+\|  \leq & \, \bar{\sigma}_g^o \|x_1^+-x_2^+\| + \bar{\sigma}_c^x\frac{1}{4}\bigg[\|W_o\|\|u_1-u_2\| + \|U_o\| \|\xi_1-\xi_2\|\bigg]\\
\leq &\, \bar{\sigma}_g^o \bigg[\bar{\sigma}_g^f \|x_1-x_2\| + \alpha \|\xi_1-\xi_2\| + \beta\|u_1-u_2\|\bigg] + \bar{\sigma}_c^x\frac{1}{4}\bigg[\|W_o\|\|u_1-u_2\| + \|U_o\| \|\xi_1-\xi_2\|\bigg]\\
\leq & \, \bar{\sigma}_g^o\bar{\sigma}_g^f \|x_1-x_2\| + \bigg[\alpha \bar{\sigma}_g^o + \frac{1}{4}\bar{\sigma}_c^x\|U_o\| \bigg] \|\xi_1-\xi_2\|+ 
	 \bigg[\beta \bar{\sigma}_g^o + \frac{1}{4}\bar{\sigma}_c^x \|W_o\| \bigg] \|u_1-u_2\| 
\end{aligned}
\end{equation}
Grouping inequalities \eqref{eq:proof:deltaiss:x1_x2} and \eqref{eq:proof:deltaiss:xi1_xi2}, we obtain that
\begin{equation}
\begin{bmatrix}
\|x_1^+-x_2^+\|\\
\|\xi_1^+-\xi_2^+\|\\
\end{bmatrix} \leq A_\delta \begin{bmatrix}
\|x_1-x_2\|\\
\|\xi_1-\xi_2\|\\
\end{bmatrix} + B_\delta \|u_1-u_2\|
\label{eq:aux_system}
\end{equation}
where $A_\delta$ is defined in \eqref{eq:deltaiss:A} and
$B_\delta=\begin{bmatrix}
\beta\\
\beta \bar{\sigma}_g^o + \frac{1}{4}\bar{\sigma}_c^x\|W_o\| \\
\end{bmatrix}$.\\
%
%We now finally show that the stability of $A$ implies the $\delta$ISS property of system \eqref{eq:model:lstm}.
%Let us first recall that $$\bigg\|\begin{bmatrix}\|x_1-x_2\|\\
%\|\xi_1-\xi_2\|\\
%\end{bmatrix}\bigg\|=\bigg\|\begin{bmatrix}x_1-x_2\\
%\xi_1-\xi_2\\
%\end{bmatrix}\bigg\|=\|\chi_1-\chi_2\|$$
%%
%Consider system \eqref{eq:aux_system}: we obtain, by iteration, that
%\begin{align}
%\begin{bmatrix}
%\|x_1(k)-x_2(k)\|\\
%\|\xi_1(k)-\xi_2(k)\| \\
%\end{bmatrix} \leq &A^k
%\begin{bmatrix}
%\|x_{01}-x_{02}\|\\
%\|\xi_{01}-\xi_{02}\|\\
%\end{bmatrix} + \sum_{i=0}^{k-1}A^{k-i-1}B\|u_1(i)-u_2(i)\|
%\label{eq:aux_system_norm}
%\end{align}
%
%By taking the norm both sides we compute
%\begin{align}
%\|\chi_1(k)-\chi_2(k)\| \leq &\|A^k\| \|\chi_{01}-\chi_{02}\| + \bigg\| \sum_{i=0}^{k-1}A^{k-i-1}B\|u_1(i)-u_2(i)\| \bigg\|
%\end{align}
%
%With standard norm arguments, and recalling that $A$ is Schur stable, there exist constants $\mu\geq 1$ and $\lambda\in(0,1)$ such that
As shown in the proof of Theorem \ref{thm:iss}, it is hence possible to write
\begin{equation*}
\|\chi_1(k)-\chi_2(k)\| \leq \mu_\delta \lambda_\delta^k \|\chi_{01}-\chi_{02}\| + \|(I_{2}-A_\delta)^{-1}B_\delta\| \max_{h\geq 0}\|{u}_1(h)-{u}_2(h)\|,
\end{equation*}
i.e. \eqref{eq:model:lstm} is $\delta$ISS according to Definition~\ref{def:deltaISS}.
\end{proof}
\medskip

\begin{proof}[Proof of Proposition \ref{prop:deltaiss:schur}]
	In view of Lemma \ref{lemma:schur}, being $a = -\bar{\sigma}_g^f - \alpha \bar{\sigma}_g^o - \frac{1}{4} \bar{\sigma}_c^x \| U_o \|$ and $b= \frac{1}{4} \bar{\sigma}_g^f \bar{\sigma}_c^x \| U_o \|$, the matrix $A_\delta$ is Schur stable -- and $\rho(A_\delta) < 1$ -- if and only if inequalities \eqref{eq:deltaiss:condition} are fulfilled.
\end{proof}

\rev{\begin{proof}[Proof of Corollary \ref{prop:deltaiss_implies_iss}]
	In the following, Corollary \ref{prop:deltaiss_implies_iss} is demonstrated by showing that the satisfaction of inequality \eqref{eq:deltaiss:condition} implies the fulfillment of \eqref{eq:iss:condition}.
	In light of the definition of $\alpha$, the left-hand inequality of \eqref{eq:deltaiss:condition} reads as
	\begin{equation}
		-1 + \bar{\sigma}_g^f + \bigg( \frac{1}{4} \| U_f \| \frac{\bar{\sigma}_g^i \bar{\sigma}_g^o}{1 - \bar{\sigma}_g^f} + \bar{\sigma}_g^i \| U_c \| + \frac{1}{4} \|U_i\| \bar{\sigma}_c^c \bigg) \bar{\sigma}_g^o + \frac{1}{4} \bar{\sigma}_c^x \| U_o \| < \frac{1}{4} \bar{\sigma}_g^f \bar{\sigma}_c^x \| U_o \|,
	\end{equation} 
	which can be re-written as
	\begin{equation}
		-1 + \bar{\sigma}_g^f + \bar{\sigma}_g^o  \bar{\sigma}_g^i \| U_c \| <
		-\frac{1}{4} (1 - \bar{\sigma}_g^f) \bar{\sigma}_c^x  \| U_o \| - \frac{1}{4} \bar{\sigma}_g^o \| U_f \| \frac{\bar{\sigma}_g^i \bar{\sigma}_g^o}{1 - \bar{\sigma}_g^f} - \frac{1}{4} \|U_i\| \bar{\sigma}_c^c \bar{\sigma}_g^o.
	\end{equation}
	Recalling that $\bar{\sigma}_g^f \in (0, 1)$, the right-hand side of the inequality is surely negative, i.e.
	\begin{equation}
	-1 + \bar{\sigma}_g^f + \bar{\sigma}_g^o  \bar{\sigma}_g^i \| U_c \| < 0.
	\end{equation}
	Condition \eqref{eq:iss:condition} is hence fulfilled.
\end{proof}}

\begin{proof}[Proof of Theorem \ref{thm:obsv_convergence}]
Let us define the error variables $e_x=x-\hat{x}$, $e_{\xi}=\xi-\hat{\xi}$ and compute their evolution over time. In particular:
\begin{equation}\label{eq:proof:obsv:dev_x}
\begin{aligned}
e_x^+=&x^+-\hat{x}^+\\=& 	\sigma_g(W_f u + U_f \xi + b_f) \circ x+ \sigma_g(W_i u + U_i \xi + b_i) \circ \sigma_c(W_c u + U_c \xi + b_c) -\\
			& 	\bigg\{\sigma_g \big[W_f u + U_f \hat{\xi} + b_f + L_f(y-\hat{y}) \big] \circ \hat{x}+ \sigma_g[W_i u + U_i \hat{\xi} + b_i + L_i(y-\hat{y}) ] \circ \sigma_c(W_c u + U_c \hat{\xi} + b_c)\bigg\} \\
=			&	\sigma_g [ W_f u + U_f \xi + b_f + L_f (y - \hat{y}) ] \circ (x-\hat{x}) +	x \circ \big[ \sigma_g(W_f u + U_f xi + b_f) - \sigma_g[W_f u + U_f \hat{\xi} + b_f + L_f(y-\hat{y}) ] \big] \\
			&	+ \sigma_g(W_i u + U_i \xi + b_i) \circ \big[ \sigma_c(W_c u + U_c \xi + b_c)-\sigma_c(W_c u + U_c \hat{\xi} + b_c) \big]\\
			&	+\sigma_c(W_c u + U_c \hat{\xi} + b_c) \circ \big[\sigma_g(W_i u + U_i \xi + b_i)-\sigma_g[W_i u + U_i \hat{\xi} + b_i + L_i(y-\hat{y}) ] \big].	
\end{aligned}
\end{equation}
Recalling \eqref{eq:obsv:sigma_g_f}-\eqref{eq:obsv:alpha_hat}, and noting  that $y- \hat{y} = C \left(\xi - \hat{\xi}\right)$, by taking the norm of both sides of \eqref{eq:proof:obsv:dev_x} we get
\begin{equation}\label{eq:proof:obsv:e_x}
\|e_x^+\|\leq \,\hat{\bar{\sigma}}_g^f \|e_x\| + \frac{\bar{\sigma}_g^i\bar{\sigma}_c^c}{1-\bar{\sigma}_g^f}\frac{1}{4}\|U_f-L_fC\|\|e_{\xi}\| + \bar{\sigma}_g^i\|U_c\|\|e_{\xi}\|+\bar{\sigma}_c^c\frac{1}{4}\|U_i-L_iC\|\|e_{\xi}\| \leq \,\hat{\bar{\sigma}}_g^f \|e_x\| + \hat{\alpha} \|e_{\xi}\|.
\end{equation}
Similarly, the evolution of $e_{\xi}$ can be computed as follows
\begin{equation}
\begin{aligned}
e_{\xi}^+=\xi^+-\hat{\xi}^+ =& \,\sigma_g(W_o u + U_o \xi + b_o) \circ \sigma_c(x^+)  -  \sigma_g[W_o u + U_o \hat{\xi} + b_o + L_o(y-\hat{y}) ] \circ \sigma_c(\hat{x}^+) \\
=& \,	\sigma_g [ W_o u + U_o \xi + b_o + L_o (y - \hat{y}) ] \circ (\sigma_c(x^+) - \sigma_c(\hat{x}^+) ) \\
& + \sigma_c(x^+) \circ	\big[\sigma_g(W_o u + U_o \xi + b_o)-\sigma_g[W_o u + U_o \hat{\xi} + b_o + L_o(y-\hat{y}) ] \big]
\end{aligned}
\end{equation}
Then, recalling \eqref{eq:obsv:sigma_g_f}, taking the norm of both sides we obtain 
\begin{equation}\label{eq:proof:obsv:e_xi}
\begin{aligned}
	\|e_{\xi}^+\| \leq & \, \hat{\bar{\sigma}}_g^o \|e_x^+\| + \frac{1}{4}\|U_o-L_oC\|\|e_{\xi}\|\bar{\sigma}_c^x 
	\leq \hat{\bar{\sigma}}_g^o \big(\hat{\bar{\sigma}}_g^f \|e_x\| + \hat{\alpha}\|e_{\xi}\|\big)+ \bar{\sigma}_c^x\frac{1}{4}\|U_o-L_oC\|\|e_{\xi}\| \\
	\leq	&\hat{\bar{\sigma}}_g^o \hat{\bar{\sigma}}_g^f \|e_x\|+\hat{(\bar{\sigma}}_g^o \hat{\alpha}+\frac{1}{4}\|U_o - L_o C\|\bar{\sigma}_c^x)\|e_{\xi}\|	
\end{aligned}
\end{equation}
Combining \eqref{eq:proof:obsv:e_x} and \eqref{eq:proof:obsv:e_xi} we can write
\begin{align}
\label{eq:eq:e_chi}
&\begin{bmatrix}
\|e_x^+\| \\
\|e_{\xi}^+\| \end{bmatrix}\leq \hat{A}\begin{bmatrix}
\|e_x\|
\\
\|e_{\xi}\|
\end{bmatrix}
\end{align}
with $\hat{A}$ being defined as in \eqref{eq:obsv:A_hat}. 
Following the same steps as the Proof of Theorem \ref{thm:iss}, if $\rho(\hat{A}) < 1$  the norm of the prediction error $\| \chi - \hat{\chi} \|$ exponentially converges to zero.
\end{proof}

\begin{proof}[Proof of Proposition \ref{prop:obsv:tuning}]
	The Schur stability of the matrix $\hat{A}$ defined in \eqref{eq:obsv:A_hat} can be assessed applying Lemma \ref{lemma:schur}, leading to the condition \eqref{eq:obsv:tuning:constraint}. 
	Such condition is applied as a constraint in the optimization problem \eqref{eq:obsv:tuning}, so as to ensure that Theorem \ref{thm:obsv_convergence} holds.
\end{proof}

\begin{proof}[Proof of Theorem \ref{thm:mpc_stabilizing}]
Assume that, at time step $k$, the optimal solution $U(k|k)$ to the MPC problem is obtained, and that according to the Receding Horizon principle the first optimal input value $u(k)=u(k|k)$ is applied to the system. 
We denote with $\chi(k+i|k)$, with $i \in \{0,\dots,N\}$, the state trajectory obtained iterating \eqref{eq:model:lstm} with initial condition $\chi(k|k)=\hat{\chi}(k)$, and using $U(k|k)$ as input sequence. Similarly, we denote $\Delta(k+N|k)=\begin{bmatrix}
\|x(k+N|k)		-\bar{x}\|\\
\|\xi(k+N|k)	-\bar{\xi}\|
\end{bmatrix}$.

We will consider the optimal value of $J$ (denoted $J^*(k)$) as a candidate Lyapunov function to analyze the stability properties of the MPC control algorithm:
	\begin{equation*}
	J^*(k)=\sum\limits_{i=0}^{N-1}\biggl(\|\chi(k+i|k)\!-\!\bar{\chi}\|_Q^2\!+\!\|u(k+i|k)\!-\!\bar{u}\|_R^2\biggr)\!\nonumber+\! \|\Delta(k+N|k)\|_P^2.
	\end{equation*}
First, it should be noted that
	\begin{equation}\label{eq:proof:mpc:ineq1}
    J^*(k)\geq \gamma_1\left\|\hat{\chi}(k)-\bar{\chi}\right\|^2,
	\end{equation}
	where $\gamma_1=\lambda_{\min}(Q)$.
	Secondly, note that $u(k+i)=\bar{u}$ is a possibly suboptimal yet feasible control input for all $i\in \{0,\dots,N-1\}$. We denote with $\chi^o(k+i|k)=[x^o(k+i|k)^T\,\,\xi^o(k+i|k)^T]^T$, with $i\in \{0,\dots,N\}$, the state trajectory obtained iterating \eqref{eq:model:lstm} with initial condition $\chi(k|k)=\hat{\chi}(k)$, and using $\{u(k)=\bar{u},\dots,u(k+N-1)=\bar{u}\}$ as input sequence.
	We thus obtain that
	\begin{equation*}
	J^*(k)\leq \sum\limits_{i=0}^{N-1}\|\chi^o(k+i|k)-\bar{\chi}\|_Q^2\nonumber+ \! \|\Delta^o(k+N|k)\|_P^2,
	\end{equation*}
    where
    $\Delta^o(k+N|k)=[\|x^o(k+N|k)-\bar{x}\|,\,\,\|\xi^o(k+N|k)-\bar{\xi}\|]^T$.
    First note that $\|\Delta^o(k+N|k)\|_P^2\leq \lambda_{\max}(P)\|\Delta^o(k+N|k)\|^2=\lambda_{\max}(P)\|\chi^o(k+N|k)-\bar{\chi}\|^2$.
    As in the proof of Theorem \ref{thm:deltaISS}, since $\rho(A_\delta)<1$ there exist $\mu\geq 0$ and $\lambda\in(0,1)$ such that, $\forall i\geq 0$,
    \begin{equation*}
	\|\chi^o(k+i|k)-\bar{\chi}\|\leq  \mu\lambda^i\|\hat{\chi}(k)-\bar{\chi}\|.
	\end{equation*}
	This implies that there exists a constant $\gamma_2 \geq 0$ such that $$J^*(k)\leq \gamma_2\|\hat{\chi}(k)-\bar{\chi}\|^2.$$
	At time $k+1$ (with some abuse of notation, but for the sake of simplicity), we denote with $\chi(k+i|k+1)$, with $i \in \{1,\dots,N+1\}$, the possibly suboptimal state trajectory obtained iterating \eqref{eq:model:lstm} with initial condition $\chi(k+1|k+1)=\hat{\chi}(k+1)$, and using $\{u(k+1|k),\dots,u(k+N-1|k),\bar{u}\}$ as input sequence. Note that $\chi(k+1|k+1)=\hat{\chi}(k+1) \neq \varphi(\hat{\chi}(k|k),u(k|k))=\chi(k+1|k)$.
	For all $i \in \{1, ..., N-1\}$, we introduce the following quantities
	\begin{equation}
	\begin{aligned}
		\varepsilon(k+1|k+1) &=\hat{\chi}(k+1)-\chi(k+1|k),\\
		\varepsilon_{\Delta}(k+1|k+1) &=\Delta(k+1|k+1)-\Delta(k+1|k),\\
		\varepsilon(k+i+1|k+1) &=\chi(k+i+1|k+1)-\chi(k+i+1|k),\\ \varepsilon_{\Delta}(k+i+1|k+1) &=\Delta(k+i+1|k+1)-\Delta(k+i+1|k).
	\end{aligned}
	\end{equation}
Then, the optimal value $J^*(k+1)$ satisfies
\begin{equation*}
		\begin{aligned}
J^*(k+1)\leq&\sum\limits_{i=0}^{N-1}\biggl(\|\chi(k+1+i|k+1)\!-\!\bar{\chi}\|_Q^2\!+\!\|u(k+1+i|k)\!-\!\bar{u}\|_R^2\biggr)\!+
	\! \|\Delta(k+N+1|k+1)\|_P^2\\
	\leq&\sum\limits_{i=0}^{N-1}\biggl(\|\chi(k+1+i|k)\!-\!\bar{\chi}\! + \!\varepsilon(k+i+1|k+1)\|_Q^2\!+\!\|u(k+1+i|k)\!-\!\bar{u}\|_R^2\biggr)\!\\
	& + \|\Delta(k+N+1|k)+\!\varepsilon_{\Delta}(k+N+1|k+1)\|_P^2
	\end{aligned}
\end{equation*}
	Therefore  
	\begin{equation}\label{eq:diseq_opt_cost_fcn}
	\begin{aligned}	 
	J^*(k+1)-J^*(k) \leq &-\|\chi(k|k)-\bar{\chi}\|^2_Q-\|u(k|k)-\bar{u}\|^2_{R} \\		 &+\sum\limits_{i=1}^{N-1}\biggl(\|\chi(k+i|k)-\bar{\chi}+\varepsilon(k+i|k+1)\|_Q^2-\|\chi(k+i|k)-\bar{\chi}\|_Q^2\biggr) \\
	& + \|\chi(k+N|k)-\bar{\chi}+\varepsilon(k+N|k+1)\|_Q^2 +\|u(k+N|k+1)-\bar{u}\|_R^2\\
	& +\|\Delta(k+N+1|k)+\varepsilon_{\Delta}(k+N+1|k+1)\|_P^2 -\! \|\Delta(k+N|k)\|_P^2	
	\end{aligned}
	\end{equation}
	We now consider the different additive terms at the right hand side of inequality \eqref{eq:diseq_opt_cost_fcn}. First, we write 
\begin{equation*}
\begin{aligned}
	& \sum\limits_{i=1}^{N-1}\biggl(\|(\chi(k+i|k)-\bar{\chi})+\varepsilon(k+i|k+1)\|_Q^2-\|\chi(k+i|k)-\bar{\chi}\|_Q^2\biggr) \\ &\quad=\sum\limits_{i=1}^{N-1}\biggl(\|\chi(k+i|k)-\bar{\chi}\|^2_Q-\|\chi(k+i|k)-\bar{\chi}\|_Q^2+\|\varepsilon(k+i|k+1)\|_Q^2 +2(\chi(k+i|k)-\bar{\chi})^TQ\varepsilon(k+i|k+1)\biggr) \\	  &\quad=\sum\limits_{i=1}^{N-1}\biggl(\|\varepsilon(k+i|k+1)\|_Q^2+2(\chi(k+i|k)-\bar{\chi})^TQ\varepsilon(k+i|k+1)\biggr).
\end{aligned}
\end{equation*}
	\normalsize
	Also, in view of Property \ref{prop:norms}, it holds that
\begin{equation*}
\begin{aligned}
	-\|&\Delta(k+N|k)\|_P^2 + \|\chi(k+N|k)-\bar{\chi}+\varepsilon(k+N|k+1)\|_Q^2 +\|\Delta(k+N+1|k)+\varepsilon(k+N+1|k+1)\|_P^2 \\
	&\quad \leq -\|\Delta(k+N|k)\|_P^2+(1+\rho^2)\|\chi(k+N|k)-\bar{\chi}\|_Q^2+\\
	&\qquad + \left(1+\frac{1}{\rho^2}\right)\|\varepsilon(k+N|k+1)\|_Q^2+(1+\rho^2)\|\Delta(k+N+1|k)\|_P^2 +\left(1+\frac{1}{\rho^2}\right)\|\varepsilon_{\Delta}(k+N+1|k+1)\|_P^2\\
	&\quad =-\|\Delta(k+N|k)\|_P^2+\left(1+\rho^2\right)\bigl(\|\chi(k+N|k)-\bar{\chi}\|_Q^2 +\|\Delta(k+N+1|k)\|_P^2\bigr) + \\	
    &\qquad +\left(1+\frac{1}{\rho^2}\right)(\|\varepsilon(k+N|k+1)\|_Q^2+\|\varepsilon_{\Delta}(k+N+1|k+1)\|_P^2).
\end{aligned}
\end{equation*}
    Noting that
    \begin{equation} \label{eq:proof:mpc:property_norm_deltaX}
    \|\Delta\| = \biggl\|\begin{bmatrix}
    \|x-\bar{x}\|\\
    \|\xi-\bar{\xi}\|\\
    \end{bmatrix}\biggr\| =
    \biggl\|\begin{bmatrix}
    x-\bar{x}\\
    \xi-\bar{\xi}\\
    \end{bmatrix}\biggr\| =
    \|\chi - \bar{\chi}\|,
    \end{equation}
    and in view of the $\delta$ISS property of the system:
    \begin{equation} \label{eq:proof:mpc:evolution_delta}
    \Delta(k+N+1|k) \leq A_\delta \,\Delta(k+N|k).
    \end{equation}
    Thus it follows that
	\begin{equation*} 
	\begin{aligned}
	-\|&\Delta(k+N|k)\|_P^2+(1+\rho^2)\bigl(\|\chi(k+N|k)-\bar{\chi}\|^2_Q+\|\Delta(k+N+1|k)\|_P^2\bigr)\\
	& \quad \leq -\|\Delta(k+N|k)\|_P^2+(1+\rho^2)q\bigl(\|\chi(k+N|k)-\bar{\chi}\|^2 +\|A_\delta \Delta(k+N|k)\|_P^2\bigr)\\
	& \quad \leq \|\Delta(k+N|k)\|^2_{A_\delta^TPA_\delta-P+qI_2(1+\rho^2)}. \\
	\end{aligned}
	\end{equation*} 
	By construction $A_\delta^TPA_\delta-P \prec -qI_2$, see \eqref{eq:mpc:P_matrix}, therefore we can always select a value of $\rho>0$ small enough to obtain $A_\delta^TPA_\delta-P+qI_2(1+\rho^2) \prec 0$. Overall, we obtain that
	\begin{equation}\label{eq:proof:mpc:ineq_opt_cost1}
	J^*(k+1)-J^*(k)\leq -\|\hat{\chi}(k)-\bar{\chi}\|^2_Q-\|{u}(k)-\bar{u}\|^2_{R} + \text{(a)},
	\end{equation}
	where
	\begin{equation*}
	\begin{aligned}
		\text{(a)}=&\sum\limits_{i=1}^{N-1}\biggl(\|\varepsilon(k+i|k+1)\|_Q^2+2(\chi(k+i|k)-\bar{\chi})^TQ\varepsilon(k+i|k+1)\biggr)\\
		&+\left(1+\frac{1}{\rho^2}\right)\left(\|\varepsilon(k+N|k+1)\|_Q^2+\|\varepsilon_{\Delta}(k+N+1|k+1)\|_P^2\right)
	\end{aligned}
	\end{equation*}
	\normalsize
Note that 
\begin{equation}
\begin{aligned}
\|\varepsilon_{\Delta}(k+N+1|k+1)\|_P^2  \leq &\lambda_{\rm\scriptscriptstyle max}(P)((\|x(k+N+1|k+1)-\bar{x}\|-\|x(k+N+1|k)-\bar{x}\|)^2 + \\
&+(\|\xi(k+N+1|k+1)-\bar{\xi}\|-\|\xi(k+N+1|k)-\bar{\xi}\|)^2) \\
\leq &\lambda_{\rm\scriptscriptstyle max}(P)(\|x(k+N+1|k+1)-x(k+N+1|k)\|^2 + \\
&+\|\xi(k+N+1|k+1)-\xi(k+N+1|k)\|^2)
\end{aligned}
\end{equation}
The latter inequality is justified by the fact that, given two vectors $a$ and $b$, $|\,\|a\|-\|b\|\,|\leq \|a-b\|$. In view of this,
$\|\varepsilon_{\Delta}(k+N+1|k+1)\|_P^2\leq \lambda_{\rm\scriptscriptstyle max}(P)\|\chi(k+N+1|k+1)-\chi(k+N+1|k)\|^2=\lambda_{\rm\scriptscriptstyle max}(P)\|\varepsilon(k+N+1|k+1)\|^2$.\\
Overall we can, for simplicity, define an upper bound to (a) as follows
\begin{align} \text{(a)}\leq&\alpha_1\sum\limits_{i=1}^{N+1}\|\varepsilon(k+i|k+1)\|^2+\alpha_2\sum\limits_{i=1}^{N-1}\|\varepsilon(k+i|k+1)\|\notag
\label{eq:proof:mpc:ineqa}
\end{align}
where $\alpha_1$ and $\alpha_2$ are suitable positive scalars. 
In the following we analyze more in details the terms $\varepsilon(k+i|k+1)$.
First, recalling \eqref{eq:model:bound:sigma_g_f}-\eqref{eq:model:bound:sigma_c_x} and \eqref{eq:obsv:sigma_g_f}-\eqref{eq:obsv:alpha_hat}, let us note that the invariant set of $\hat{x}_{(j)}$ is $\hat{\mathcal{X}}  = \bigg\{ x \in \mathbb{R} : | x | \leq \frac{\hat{\bar{\sigma}}_g^i \hat{\bar{\sigma}}_g^o}{1- \hat{\bar{\sigma}}_g^f}  \bigg\}$, and thus $\hat{\bar{\sigma}}_c^x = \sigma_c \bigg( \frac{\hat{\bar{\sigma}}_g^i \hat{\bar{\sigma}}_g^o}{1- \hat{\bar{\sigma}}_g^f} \bigg)$. 
The two sub-vectors of  $\varepsilon(k+1|k+1) = [ (\hat{x}(k+1) - x(k+1|k))^T, (\hat{\xi}(k+1) - \xi(k+1 |k))^T ]^T$ are now computed as: \\[1mm]
\begin{equation} \label{eq:proof:mpc:epsilon_x}
\begin{aligned}
	\hat{x}(k+1) - x(k+1|k) =& \sigma_g[W_f u + U_f \hat{\xi} + b_f + L_f C (\xi - \hat{\xi})] \circ \hat{x} + \sigma_g[W_i u + U_i \hat{\xi} + b_i + L_i C (\xi - \hat{\xi})] \circ \sigma_c[W_c u + U_c \hat{\xi} + b_c] \\
	& -\sigma_g[W_f u + U_f \hat{\xi} + b_f] \circ \hat{x} + \sigma_g[W_i u + U_i \hat{\xi} + b_i] \circ \sigma_c[W_c u + U_c \hat{\xi} + b_c] 
\end{aligned}
\end{equation} \vspace{-2mm}
\begin{equation}\label{eq:proof:mpc:epsilon_xi}
\begin{aligned}
\hat{\xi}(k+1) - \xi(k+1|k) =& \sigma_g[W_o u + U_o \hat{\xi} + b_o + L_o C (\xi - \hat{\xi})] \circ \sigma_c(\hat{x}^+) - \sigma_g[W_o u + U_o \hat{\xi} + b_o ] \circ \sigma_c(x(k+1|k)) 
\\
=& \big[ \sigma_g[W_o u + U_o \hat{\xi} + b_o + L_o C (\xi - \hat{\xi})] - \sigma_g[W_o u + U_o \hat{\xi} + b_o + L_o C (\xi - \hat{\xi})]  \big] \circ  \sigma_c(\hat{x}^+) + \\
&+ \sigma_g[W_o u + U_o \hat{\xi} + b_o + L_o C (\xi - \hat{\xi})] \big( \sigma_c(\hat{x}^+) - \sigma_c(x(k+1|k))\big)
\end{aligned}
\end{equation}\\[1mm]
Thus, taking the norm of both sides of \eqref{eq:proof:mpc:epsilon_x} and \eqref{eq:proof:mpc:epsilon_xi}, and exploiting the aforementioned bounds, it can be shown that
\begin{subequations}
	\begin{equation}
	\| \varepsilon(k+1|k+1) \| \leq \alpha_3
	\|\chi(k)-\hat{\chi}(k)\|,
	\end{equation}
	where 
	\begin{equation}
		\alpha_3=\left\|\,\, \begin{bmatrix}
		0 & \quad \frac{1}{4}\frac{\hat{\bar{\sigma}}_g^i \hat{\bar{\sigma}}_c^c}{1-\hat{\bar{\sigma}}^f_g}\|L_fC\| + \frac{1}{4}\bar{\sigma}_c^c\|L_iC\|\\
		0 & \quad \bar{\sigma}_g^o \bigg( \frac{1}{4}\frac{\hat{\bar{\sigma}}_g^i \hat{\bar{\sigma}}_c^c}{1-\hat{\bar{\sigma}}^f_g} \|L_fC\| + \frac{1}{4}\bar{\sigma}_c^c\|L_iC\| \bigg) +\frac{1}{4}\hat{\bar{\sigma}}_c^x\|L_oC\|
		\end{bmatrix}  \,\, \right\|.
	\end{equation}
\end{subequations}\\[1mm]
Letting $e(k)=\chi(k)-\hat{\chi}(k)$, in view of the error convergence rate ensured by the observer, see Theorem \ref{thm:obsv_convergence}, we can guarantee that
	\begin{equation}
	\|\varepsilon(k+1|k+1)\|\leq \alpha_3 \mu \lambda^{k}\|e(0)\|.
	\end{equation}
	Thus, in view of $\delta$ISS of the system (see Theorem \ref{thm:deltaISS}) we can further state that
	\begin{equation}
	\|\varepsilon(k+i+1|k+1)\|\leq \mu \lambda^i \|\varepsilon(k+1|k+1)\|
	\end{equation}\normalsize
	This implies that
	\begin{equation}
	\begin{array}{c}
	\|\varepsilon(k+i|k+1)\|\leq \alpha_3 \mu^2 \lambda^{k+i-1}\|e(0)\|
	\label{eq: eps}
	\end{array}
	\end{equation}
	By combining inequalities \eqref{eq:proof:mpc:ineq_opt_cost1} and \eqref{eq: eps} we eventually obtain that there exists a $\mathcal{KL}$ function $\tilde{\beta}$ and a constant $\gamma_3=\gamma_1>0$ such that $\tilde{J}^*(k+1)-\tilde{J}^*(k)\leq-\gamma_3\left\|\hat{\chi}(k)-\bar{\chi}\right\|^2
	+\tilde{\beta}(\|e(0)\|,k)$, where $\tilde{\beta}(\|e(0)\|,k)$ is exponentially decreasing with respect to its second argument $k$.
Along the same lines of reasoning of Theorem 3 in \cite{Scokaert97}, we can prove asymptotic stability of the equilibrium point denoted by the triplet $(\bar{u},\bar{\chi},\bar{y})$.

\end{proof}

\bibliography{mybiblio}

\begin{thebibliography}{10}

\bibitem{alessandri2008moving}
Angelo Alessandri, Marco Baglietto, and Giorgio Battistelli.
\newblock Moving-horizon state estimation for nonlinear discrete-time systems:
  New stability results and approximation schemes.
\newblock {\em Automatica}, 44(7):1753--1765, 2008.

\bibitem{amrouche2018long}
Massinissa Amrouche, Deka~Shankar Anand, Aleksandra Leki{\'c},
  Vicen{\c{c}}~Rubies Royo, Elaina~Teresa Chai, Du{\v{s}}an~M Stipanovi{\'c},
  Boris Murmann, and Claire~J Tomlin.
\newblock Long short-term memory neural network equilibria computation and
  analysis.
\newblock NIPS 2018 Workshop Spatiotemporal Blind Submission, 2018.

\bibitem{8728229}
L.~B. {Armenio}, E.~{Terzi}, M.~{Farina}, and R.~{Scattolini}.
\newblock Model predictive control design for dynamical systems learned by echo
  state networks.
\newblock {\em IEEE Control Systems Letters}, 3(4):1044--1049, Oct 2019.

\bibitem{aswani2013provably}
Anil Aswani, Humberto Gonzalez, S~Shankar Sastry, and Claire Tomlin.
\newblock Provably safe and robust learning-based model predictive control.
\newblock {\em Automatica}, 49(5):1216--1226, 2013.

\bibitem{Discrete-timedeltaISS}
F.~Bayer, M.~B{\"u}rger, and F.~Allg{\"o}wer.
\newblock Discrete-time incremental {ISS}: A framework for robust {NMPC}.
\newblock In {\em Control Conference (ECC), 2013 European}, pages 2068--2073.
  IEEE, European Control Conference (ECC), 2013.

\bibitem{bonassi2019lstm}
Fabio Bonassi, Enrico Terzi, Marcello Farina, and Riccardo Scattolini.
\newblock Lstm neural networks: Input to state stability and probabilistic
  safety verification.
\newblock Learning for Dynamics and Control (L4DC), 2020.
\newblock arXiv:1912.04377.

\bibitem{bristow2006survey}
Douglas~A Bristow, Marina Tharayil, and Andrew~G Alleyne.
\newblock A survey of iterative learning control.
\newblock {\em IEEE control systems magazine}, 26(3):96--114, 2006.

\bibitem{campi2009scenario}
Marco~C Campi, Simone Garatti, and Maria Prandini.
\newblock The scenario approach for systems and control design.
\newblock {\em Annual Reviews in Control}, 33(2):149--157, 2009.

\bibitem{CAMPI20021337}
M.C. Campi, A.~Lecchini, and S.M. Savaresi.
\newblock Virtual reference feedback tuning: a direct method for the design of
  feedback controllers.
\newblock {\em Automatica}, 38(8):1337 -- 1346, 2002.

\bibitem{8702629}
S.~A. {Deka}, D.~M. {Stipanović}, B.~{Murmann}, and C.~J. {Tomlin}.
\newblock Long-short term memory neural network stability and stabilization
  using linear matrix inequalities.
\newblock pages 1--4. IEEE International Symposium on Circuits and Systems
  (ISCAS), May 2019.

\bibitem{Deka2019}
Shankar~A. Deka, Du{\v{s}}an~M. Stipanovi{\'{c}}, Boris Murmann, and Claire~J.
  Tomlin.
\newblock Global asymptotic stability and stabilization of long short-term
  memory neural networks with constant weights and biases.
\newblock {\em Journal of Optimization Theory and Applications},
  181(1):231--243, Apr 2019.

\bibitem{delgado1995dynamic}
A~Delgado, C~Kambhampati, and Kevin Warwick.
\newblock Dynamic recurrent neural network for system identification and
  control.
\newblock {\em IEE Proceedings-Control Theory and Applications},
  142(4):307--314, 1995.

\bibitem{falugi2013getting}
Paola Falugi and David~Q Mayne.
\newblock Getting robustness against unstructured uncertainty: a tube-based mpc
  approach.
\newblock {\em IEEE Transactions on Automatic Control}, 59(5):1290--1295, 2013.

\bibitem{fazlyab2019efficient}
Mahyar Fazlyab, Alexander Robey, Hamed Hassani, Manfred Morari, and George
  Pappas.
\newblock Efficient and accurate estimation of lipschitz constants for deep
  neural networks.
\newblock pages 11427--11438. Advances in Neural Information Processing
  Systems, 2019.

\bibitem{gers2001lstm}
Felix~A Gers and E~Schmidhuber.
\newblock Lstm recurrent networks learn simple context-free and
  context-sensitive languages.
\newblock {\em IEEE Transactions on Neural Networks}, 12(6):1333--1340, 2001.

\bibitem{gers2000learning}
Felix~A Gers, J{\"u}rgen Schmidhuber, and Fred Cummins.
\newblock Learning to forget: Continual prediction with lstm.
\newblock {\em Neural Computation}, 12(10):2451--2471, 2000.

\bibitem{gers2002learning}
Felix~A Gers, Nicol~N Schraudolph, and J{\"u}rgen Schmidhuber.
\newblock Learning precise timing with lstm recurrent networks.
\newblock {\em Journal of machine learning research}, 3(Aug):115--143, 2002.

\bibitem{goodfellow2016deep}
Ian Goodfellow, Yoshua Bengio, and Aaron Courville.
\newblock {\em Deep learning}.
\newblock MIT press, 2016.

\bibitem{graves2013speech}
Alex Graves, Abdel-rahman Mohamed, and Geoffrey Hinton.
\newblock Speech recognition with deep recurrent neural networks.
\newblock pages 6645--6649. IEEE international conference on acoustics, speech
  and signal processing, 2013.

\bibitem{graves2009offline}
Alex Graves and J{\"u}rgen Schmidhuber.
\newblock Offline handwriting recognition with multidimensional recurrent
  neural networks.
\newblock pages 545--552. Advances in neural information processing systems,
  2009.

\bibitem{greff2016lstm}
Klaus Greff, Rupesh~K Srivastava, Jan Koutn{\'\i}k, Bas~R Steunebrink, and
  J{\"u}rgen Schmidhuber.
\newblock Lstm: A search space odyssey.
\newblock {\em IEEE transactions on neural networks and learning systems},
  28(10):2222--2232, 2016.

\bibitem{hall1989modelling}
Raymond~C Hall and Dale~E Seborg.
\newblock Modelling and self-tuning control of a multivariable ph
  neutralization process part i: Modelling and multiloop control.
\newblock pages 1822--1827. IEEE, American Control Conference, 1989.

\bibitem{hand2006data}
David~J Hand.
\newblock Data mining.
\newblock {\em Encyclopedia of Environmetrics}, 2, 2006.

\bibitem{haykin1994neural}
Simon Haykin.
\newblock {\em Neural networks: a comprehensive foundation}.
\newblock Prentice Hall PTR, 1994.

\bibitem{hewing2019scenario}
Lukas Hewing and Melanie~N Zeilinger.
\newblock Scenario-based probabilistic reachable sets for recursively feasible
  stochastic model predictive control.
\newblock {\em IEEE Control Systems Letters}, 4(2):450--455, 2019.

\bibitem{hippert2001neural}
Henrique~Steinherz Hippert, Carlos~Eduardo Pedreira, and Reinaldo~Castro Souza.
\newblock Neural networks for short-term load forecasting: A review and
  evaluation.
\newblock {\em IEEE Transactions on power systems}, 16(1):44--55, 2001.

\bibitem{hochreiter1998vanishing}
Sepp Hochreiter.
\newblock The vanishing gradient problem during learning recurrent neural nets
  and problem solutions.
\newblock {\em International Journal of Uncertainty, Fuzziness and
  Knowledge-Based Systems}, 6(02):107--116, 1998.

\bibitem{hochreiter1997long}
Sepp Hochreiter and J{\"u}rgen Schmidhuber.
\newblock Long short-term memory.
\newblock {\em Neural computation}, 9(8):1735--1780, 1997.

\bibitem{hou2013model}
Zhong-Sheng Hou and Zhuo Wang.
\newblock From model-based control to data-driven control: Survey,
  classification and perspective.
\newblock {\em Information Sciences}, 235:3--35, 2013.

\bibitem{jaeger2002tutorial}
Herbert Jaeger.
\newblock {\em Tutorial on training recurrent neural networks, covering BPPT,
  RTRL, EKF and the" echo state network" approach}, volume~5.
\newblock GMD-Forschungszentrum Informationstechnik Bonn, 2002.

\bibitem{jiang2001input}
Zhong-Ping Jiang and Yuan Wang.
\newblock Input-to-state stability for discrete-time nonlinear systems.
\newblock {\em Automatica}, 37(6):857--869, 2001.

\bibitem{jin2018robot}
Long Jin, Shuai Li, Jiguo Yu, and Jinbo He.
\newblock Robot manipulator control using neural networks: A survey.
\newblock {\em Neurocomputing}, 285:23--34, 2018.

\bibitem{jury1962simplified}
EI~Jury.
\newblock A simplified stability criterion for linear discrete systems.
\newblock {\em Proceedings of the IRE}, 50(6):1493--1500, 1962.

\bibitem{kohler2019simple}
Johannes K{\"o}hler, Frank Allg{\"o}wer, and Matthias~A M{\"u}ller.
\newblock A simple framework for nonlinear robust output-feedback mpc.
\newblock pages 793--798. IEEE, 18th European Control Conference (ECC), 2019.

\bibitem{kohler2020computationally}
Johannes K{\"o}hler, Raffaele Soloperto, Matthias~A Muller, and Frank Allgower.
\newblock A computationally efficient robust model predictive control framework
  for uncertain nonlinear systems.
\newblock {\em IEEE Transactions on Automatic Control}, 2020.

\bibitem{krizhevsky2012imagenet}
Alex Krizhevsky, Ilya Sutskever, and Geoffrey~E Hinton.
\newblock Imagenet classification with deep convolutional neural networks.
\newblock pages 1097--1105. Advances in neural information processing systems,
  2012.

\bibitem{lanzetti2019recurrent}
Nicolas Lanzetti, Ying~Zhao Lian, Andrea Cortinovis, Luis Dominguez, Mehmet
  Mercang{\"o}z, and Colin Jones.
\newblock Recurrent neural network based mpc for process industries.
\newblock In {\em 2019 18th European Control Conference (ECC)}, pages
  1005--1010. IEEE, 2019.

\bibitem{lenz2015deepmpc}
Ian Lenz, Ross~A Knepper, and Ashutosh Saxena.
\newblock Deepmpc: Learning deep latent features for model predictive control.
\newblock In {\em Robotics: Science and Systems}. Rome, Italy, 2015.

\bibitem{li2016distributed}
Shuai Li, Jinbo He, Yangming Li, and Muhammad~Usman Rafique.
\newblock Distributed recurrent neural networks for cooperative control of
  manipulators: A game-theoretic perspective.
\newblock {\em IEEE transactions on neural networks and learning systems},
  28(2):415--426, 2016.

\bibitem{ljung2001system}
Lennart Ljung.
\newblock System identification.
\newblock {\em Wiley Encyclopedia of Electrical and Electronics Engineering},
  2001.

\bibitem{mayne2011tube}
David~Q Mayne, Erric~C Kerrigan, EJ~Van~Wyk, and Paola Falugi.
\newblock Tube-based robust nonlinear model predictive control.
\newblock {\em International Journal of Robust and Nonlinear Control},
  21(11):1341--1353, 2011.

\bibitem{miller1995neural}
W~Thomas Miller, Paul~J Werbos, and Richard~S Sutton.
\newblock {\em Neural networks for control}.
\newblock MIT press, 1995.

\bibitem{Scokaert97}
Scokaert P.O.M., J.B. Rawlings, and E.S. Meadows.
\newblock Discrete-time stability with perturbations: application to model
  predictive control.
\newblock {\em Automatica}, 33(3):463--470, 1997.

\bibitem{rao2003constrained}
Christopher~V Rao, James~B Rawlings, and David~Q Mayne.
\newblock Constrained state estimation for nonlinear discrete-time systems:
  Stability and moving horizon approximations.
\newblock {\em IEEE transactions on automatic control}, 48(2):246--258, 2003.

\bibitem{sohrab2003basic}
Houshang~H Sohrab.
\newblock {\em Basic real analysis}, volume 231.
\newblock Springer, 2003.

\bibitem{sundermeyer2012lstm}
Martin Sundermeyer, Ralf Schl{\"u}ter, and Hermann Ney.
\newblock Lstm neural networks for language modeling.
\newblock In {\em Thirteenth annual conference of the international speech
  communication association}, 2012.

\bibitem{tanaskovic2017data}
Marko Tanaskovic, Lorenzo Fagiano, Carlo Novara, and Manfred Morari.
\newblock Data-driven control of nonlinear systems: An on-line direct approach.
\newblock {\em Automatica}, 75:1--10, 2017.

\bibitem{enrico_terzi_2020_3956067}
Enrico Terzi, Fabio Bonassi, Marcello Farina, and Riccardo Scattolini.
\newblock {pH} reactor dataset, available at
  https://doi.org/10.5281/zenodo.3956067, July 2020.

\bibitem{wong2018recurrent}
Wee Wong, Ewan Chee, Jiali Li, and Xiaonan Wang.
\newblock Recurrent neural network-based model predictive control for
  continuous pharmaceutical manufacturing.
\newblock {\em Mathematics}, 6(11):242, 2018.

\bibitem{wu2008top}
Xindong Wu, Vipin Kumar, J~Ross Quinlan, Joydeep Ghosh, Qiang Yang, Hiroshi
  Motoda, Geoffrey~J McLachlan, Angus Ng, Bing Liu, S~Yu Philip, et~al.
\newblock Top 10 algorithms in data mining.
\newblock {\em Knowledge and information systems}, 14(1):1--37, 2008.

\bibitem{xingjian2015convolutional}
SHI Xingjian, Zhourong Chen, Hao Wang, Dit-Yan Yeung, Wai-Kin Wong, and
  Wang-chun Woo.
\newblock Convolutional lstm network: A machine learning approach for
  precipitation nowcasting.
\newblock pages 802--810. Advances in neural information processing systems,
  2015.

\bibitem{zhang2003time}
G~Peter Zhang.
\newblock Time series forecasting using a hybrid arima and neural network
  model.
\newblock {\em Neurocomputing}, 50:159--175, 2003.

\bibitem{zhang1998forecasting}
Guoqiang Zhang, B~Eddy Patuwo, and Michael~Y Hu.
\newblock Forecasting with artificial neural networks:: The state of the art.
\newblock {\em International journal of forecasting}, 14(1):35--62, 1998.

\end{thebibliography}
\bibliographystyle{plain}
\end{document}